%% file: paper.tex
\documentclass[a4paper]{article}

\pdfoutput=1

\usepackage[ruled,linesnumbered,algo2e]{algorithm2e}
\usepackage{amsfonts}
\usepackage{amsmath}
\usepackage{amssymb}
\usepackage{booktabs}
\usepackage{graphicx}
\usepackage{hyperref}
	\usepackage[capitalise]{cleveref}
\usepackage{listings}
\usepackage{mathbbol}
\usepackage{multirow}
\usepackage{paralist}
\usepackage{rotating}
\usepackage{sidecap}
\usepackage{stackrel}
\usepackage[subrefformat=parens,labelformat=parens]{subfig}
\usepackage{tikz}
\usepackage{times}
\usepackage{url}

\usetikzlibrary{arrows,automata}
\usetikzlibrary{positioning,decorations.pathreplacing}

\input{header}

\begin{document}

\title{Heuristics Miners for Streaming Event Data}
\author{Andrea Burattin \footnote{Email: \emph{burattin@math.unipd.it}. Affiliation: Department of Mathematics, University of Padua, Italy.} \and
Alessandro Sperduti \footnote{Email: \emph{sperduti@math.unipd.it}. Affiliation: Department of Mathematics, University of Padua, Italy.} \and
Wil M. P. van der Aalst \footnote{Email: \emph{w.m.p.v.d.aalst@tue.nl}. Affiliation: Department of Mathematics and Computer Science, Eindhoven University of Technology, The Netherlands.}}
\date{}

\maketitle

\begin{abstract}
More and more business activities are performed using information systems. These systems produce such huge amounts of event data  that existing systems are unable to store and process them. Moreover, few processes are in steady-state and due to changing circumstances processes evolve and systems need to adapt continuously. Since conventional process discovery algorithms have been defined for batch processing, it is difficult to apply them in such evolving environments. Existing algorithms cannot cope with streaming event data and tend to generate unreliable and obsolete results.

In this paper, we discuss the peculiarities of dealing with streaming event data in the context of process mining. Subsequently, we present a general framework for defining process mining algorithms in settings where it is impossible to store all events over an extended period or where processes evolve while being analyzed. We show how the Heuristics Miner, one of the most effective process discovery algorithms for practical applications, can be modified using this framework. Different stream-aware versions of the Heuristics Miner are defined and implemented in ProM. Moreover, experimental results on artificial and real logs are reported.

\end{abstract}

\input{section-introduction}
\input{section-related}
\input{section-basics}
\input{section-approach}

\input{section-implementation}
\input{section-results}
\input{section-conclusions}

\bibliographystyle{plain}
\bibliography{library}

\appendix
\input{section-heuristics-miner}
\input{section-error-bounds}

\end{document}

%% file: header.tex
\newtheorem{definition}{Definition}
\newtheorem{theorem}{Theorem}
\newtheorem{proof}{Proof}

\definecolor{attributeColor}{rgb}{0,0.5,0}
\definecolor{stringColor}{rgb}{0.7,0,0}
\definecolor{commentColor}{rgb}{0.4,0.4,0.4}
\definecolor{attributeColor}{rgb}{0.2,0.6,0.8}

\lstdefinelanguage{OpenXES}{
	morekeywords={log,trace,event,date,string},
	sensitive=false,
	morestring=[b]"
}

\crefname{lstlisting}{Listing}{listings}
\Crefname{lstlisting}{Listing}{Listings}

\newcommand{\lstsetOpenXES}{
	\lstset{%
		language=OpenXES,                  
		basicstyle=\footnotesize\ttfamily, 
		numbers=left,                   
		numberstyle=\footnotesize,      
		stepnumber=1,                   
		numbersep=10pt,                 
		backgroundcolor=\color{white},  
		showspaces=false,               
		showstringspaces=false,         
		showtabs=false,                 
		frame=none,	                    
		tabsize=2,	                    
		captionpos=b,                   
		breaklines=true,                
		breakatwhitespace=true,         
		escapeinside={\%*}{*)},         
		keywordstyle=\color{blue},
		emphstyle=\color{attributeColor}\bfseries,
		stringstyle=\color{stringColor},
		commentstyle=\color{commentColor},
		emph={key, value}
	}
}

%% file: section-introduction.tex
\section{Introduction}

One of the main aims of \emph{process mining} is control-flow discovery, i.e., learning process models from example traces recorded in some event log. Many different control-flow discovery algorithms have been proposed in the past (see \cite{VanderAalst2011}). Basically, all such algorithms have been defined for batch processing, i.e., a complete event log containing all executed activities is supposed to be available at the moment of execution of the mining algorithm. Nowadays, however, the information systems supporting business processes are able to produce a huge amount of events thus creating new opportunities and challenges from a computational point of view. In fact, in case of streaming data it may be impossible to store all events.
Moreover, even if one is able to store all event data, it is often impossible to process them due to the exponential nature of most algorithms. In addition to that, a business process may evolve over time. Manyika et al. \cite{Manyika2011} report possible ways for exploiting large amount of data to improve the company business. In their paper, \emph{stream processing} is defined as ``\emph{technologies designed to process large real-time streams of event data}'' and one of the example applications is \emph{process monitoring}.
The challenge to deal with streaming event data is also discussed in the Process Mining Manifesto\footnote{The Process Mining Manifesto is authored by the IEEE Task Force on Process Mining (\url{www.win.tue.nl/ieeetfpm/}).} \cite{PMMAnifesto}.

Currently, however, there are no process mining algorithms able to mine an event stream. This paper is the first that presents algorithms for discovering process models based on streaming event data. In the remainder of this paper we refer to this problem as \emph{Streaming Process Discovery} (or SPD).

According to \cite{Aggarwal2007,Bifet2010}, a data stream consists of an unbounded sequence of data items with a very high throughput. In addition to that, the following assumptions are typically made:
\begin{inparaenum}[\itshape i)]
	\item data is assumed to have a small and fixed number of attributes;
	\item mining algorithms should be able to process an infinite amount of data, without exceeding memory limits or otherwise fail, no matter how many items are processed;
	\item for classification tasks, data has a limited number of possible class labels;
	\item the amount of memory available to a learning/mining algorithm is considered finite, and typically much smaller than the data observed in a reasonable span of time;
	\item there is a small upper bound on the time allowed to process an item, e.g. algorithms have to scale linearly with the number of processed items: typically the algorithms work with one pass of the data; and
	\item stream ``concepts'' are assumed to be stationary or evolving \cite{VanLeeuwen2008a,Widmer1996}.
\end{inparaenum}

In SPD, a typical task is to reconstruct a control-flow model that could have generated the observed event log.
The general representation of the SPD problem that we adopt in this paper is shown in
\cref{fig:general-idea}: one or more sources emit events (represented as solid dots) which are observed by the stream miner that keeps the representation of the process model up-to-date.
Obviously, no standard  mining algorithm adopting a batch approach is able to deal with this scenario.

An SPD algorithm has to give satisfactory answers to the following two categories of questions:
\begin{enumerate}
	\item Is it possible to discover a process model while storing a minimal amount of information? What should be stored? What is the performance of such methods both in terms of model quality and speed/memory usage?
	\item Can SPD techniques deal with changing processes? What is the performance when the stream exhibits certain types of concept drift?
\end{enumerate}

\begin{figure}[t]\centering
	\includegraphics[width=\textwidth]{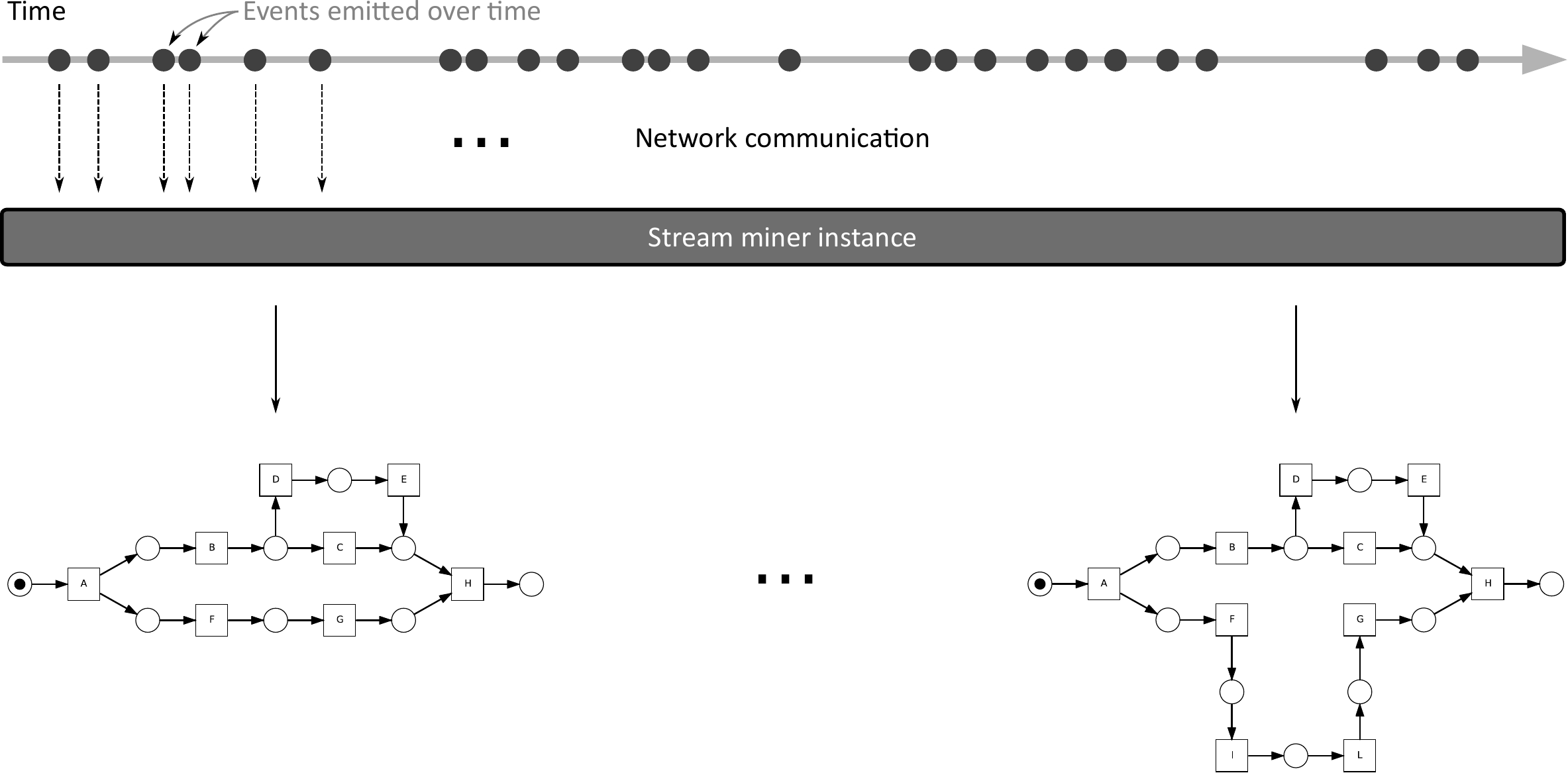}
	\caption{General idea of SPD: the stream miner continuously receives events and, using the latest observations, updates the process model.}
	\label{fig:general-idea}
\end{figure}

In this paper, we discuss the peculiarities  of mining a stream of logs in the context of process mining. Subsequently, we present a general framework for defining process mining algorithms for streams of logs. We show how the Heuristics Miner, one of the  more effective algorithms for practical applications of process mining, can be adapted for stream mining according to our SPD framework.

%% file: section-related.tex
\vspace{1em}

A data stream is defined as a ``real-time, continuous, ordered sequence of items'' \cite{Golab2003}. The ordering of the data items is expressed implicitly by the arrival timestamp of each item. Algorithms that are supposed to interact with data streams must respect some requirements, such as:
\begin{inparaenum}[\itshape a\upshape)]
	\item it is impossible to store the complete stream;
	\item backtracking over a data stream is not feasible, so algorithms are required to make only one pass over data;
	\item it is important to quickly adapt the model to cope with unusual data values;
	\item the approach must deal with variable system conditions, such as fluctuating stream rates.
\end{inparaenum}
Due to these requirements, algorithms for data streams mining are divided into two categories: data and task based \cite{Gaber2005}. The idea of the first ones is to use only a fragment of the entire dataset (by reducing the data into a smaller representation). The idea of the latter approach is to modify existing techniques (or invent new ones) to achieve time and space efficient solutions.

The main ``data based'' techniques are: sampling, load shedding, sketching and aggregation. All these are based on the idea of randomly select items or stream portions. The main drawback is that, since the dataset size is unknown, it is hard to define the number of items to collect; moreover it is possible that some of the items that are ignored were actually interesting and meaningful.
Other approaches, like aggregation, are slightly different: they are based on summarization techniques and, in this case, the idea is to consider measures such as mean and variance; with these approaches, problems arise when the data distribution contains many fluctuations.

The main ``task based'' techniques are: approximation algorithms, sliding window and algorithm output granularity.
Approximation algorithms aim to extract an approximate solution. It is possible to define error bounds on the procedure. This way, one obtains an ``accuracy measure''.
The basic idea of sliding window is that users are more interested in most recent data, thus the analysis is performed giving more importance to recent data, and considering only summarization of the old ones.
The main characteristic of ``algorithm output granularity'' is the ability to adapt the analysis to resource availability.


The task of mining data stream is typically focused on specific types of algorithms \cite{Gaber2005,Widmer1996,Aggarwal2007}. In particular, there are techniques for: clustering; classification; frequency counting; time series analysis and change diagnosis (concept drift detection). All these techniques cope with very specific problems and cannot be adapted to the SPD problem. However, as this work presents, it is possible to reuse some principles or to reduce the SPD to sub-problems that can be solved with the available algorithms.


Over the last decade dozens of process discovery techniques have been proposed \cite{VanderAalst2011}, e.g., the Heuristics Miner \cite{Weijters2003}. However, these all work on a full event log and not streaming data.
Few works in process mining literature touch issues related to mining event data streams.

In \cite{Kindler2005a,Kindler2006c}, the authors focus on incremental workflow mining and \emph{task mining} (i.e. the identification of the activities starting from the documents accessed by users).
The basic idea is to mine process instances as soon as they are observed; each new model is then merged with the previous one so to refine the global process representation.
The approach described is thought to deal with the incremental process refinement based on logs generated from version management systems. However, as authors state, only the initial idea is sketched.

An approach for mining legacy systems is described in \cite{Kalsing2010}. In particular, after the introduction of monitoring statements into the legacy code, an incremental process mining approach is presented. The idea is to apply the same heuristics of the Heuristics Miner into the process instances and add these data into an AVL tree, which are used to find the best holding relations. Actually, this technique operates on ``log fragments'' and not on single events so it is not really suitable for an online setting.
Moreover, heuristics are based on frequencies, so they must be computed with respect to a set of traces and, again, this is not suitable for the settings with streaming event data.

An interesting contribution to the analysis of evolving processes is given in the paper by Bose et al. \cite{Bose2011}. The proposed approach, based on statistical hypothesis tests, aims at detecting \emph{concept drift}, i.e. the changes in event logs, and identifying the regions of change in a process.

Sol\'e and Carmona, in \cite{Sole2012}, describe an incremental approach for translating transition systems into Petri nets. This translation is performed using Region Theory. The approach solves the problem of complexity of the translation, by splitting the log into several parts; applying the Region Theory to each of them and then combine all them. These regions are finally converted into Petri net.

The above review of the literature shows there no process mining technique for SPD that address the requirements listed in this section.

The remainder of this paper is organized as follows: \cref{sec:basic-concepts} presents the basic concepts related to SPD; \cref{sec:approach} describes the new algorithms designed to tackle stream process mining; \cref{sec:implementation} reports some details about the implementation of all the approaches in ProM and \cref{sec:results} presents the results of several experiments; \cref{sec:conclusions} concludes the paper. This work contains two appendices: \cref{appendix:HM} summarizes the Heuristics Miner algorithm, \cref{sec:error-bound} presents some details on error bounds.

%% file: section-basics.tex
\section{Basic concepts} \label{sec:basic-concepts}

The main difference between classical process mining \cite{VanderAalst2011} and SPD lies in the assumed input format. For SPD we assume streaming event data that may even come from multiple sources rather that a static event log
containing historic data.

In this paper, we assume that each \emph{event}, received by the miner, contains the \emph{name of the activity} executed, the \emph{case id} it belongs to, and a \emph{timestamp}. A formal definition of these elements is as follows:
\begin{definition}[Activity, Case, Time and Event Stream]
	Let $\mathcal{A}$ be a set of activities and $\mathcal{C}$ be a set of case identifiers. An \emph{event} is a triplet $(c,a,t) \in \mathcal{C} \times \mathcal{A} \times \mathbb{N}$, i.e., the occurrence of activity $a$ for case $c$ (i.e. the process instance) at time $t$ (timestamp of emission of the event).
	Actually, in the miner, rather than using an absolute timestamp, we consider a progressive number representing the number of events seen so far, so an event at time $t$ is followed by another event at time $t+1$, regardless the time lasts between them.  $S \in (\mathcal{C}\times \mathcal{A} \times \mathbb{N})^*$ is an event stream, i.e., a sequence of events that are observed item by item.
	The events in $S$ are sorted according to the order they are emitted, i.e. the event timestamp.
\end{definition}

Starting from this definition, it is possible to define some functions:
\begin{definition}[Case time scope]
	$t_{\text{start}}(c)=\min_{(c,a,t)\in S} t$, i.e. the time when the first activity for $c$ is observed.
	$t_{\text{end}}(c)=\max_{(c,a,t)\in S}t$, i.e. the time when the last activity for $c$ is observed.
\end{definition}

\begin{definition}[Subsequence]
	Given a sequence of events $S \in (\mathcal{C}\times \mathcal{A} \times \mathbb{N})^*$, it is a sorted series of events: $S = \langle \dots, s_i, \dots, s_{i+j}, \dots \rangle$ where $s_i = (c,a,t) \in \mathcal{C}\times \mathcal{A} \times \mathbb{N}$.
	A subsequence $S_i^j$ of $S$ is a sequence that identifies the elements of $S$ starting at position $i$ and finishing at position $i+j$: $S_i^j = \langle s_i, \dots, s_{i+j}\rangle$.
\end{definition}

In order to relate classical control-flow discovery algorithms with new algorithms for streams, we can consider an \emph{observation period}. An observation period $O$ for an event stream $S$, is a finite subsequence of $S$ starting at time $i$ and with size $j$: $O = S_i^j$. Basically, any observation period is a finite subsequence of a stream, and it can be understood as a classical log file (although the ``head'' and ``tail'' of some cases may be missing).
A well-established control-flow discovery algorithm that can be applied to an observation period log is the Heuristics Miner, whose main features are reported in \cref{appendix:HM}.

\vspace{1em}
In analogy with classical data streams, an event stream can be defined as {\it stationary} or {\it evolving}. In our context, a stationary stream can be seen as generated by a business process that does not change with time. On the contrary, an evolving stream can be understood as generated by a process that changes in time. More precisely, different modes of change can be considered: {\it i)} drift of the process model;  {\it ii)} shift of the process model;  {\it iii)} cases (i.e., execution instances of the process) distribution change.
Drift and shift of the process model correspond to the classical two modes of \emph{concept drift} \cite{Bose2011} in data streams: a drift of the model refers to a  gradual change of the underlying process, while a model shift happens when a change between two process models is more abrupt.
The change in cases distribution represents another way in which an event stream can evolve, i.e. the original process may stay the same during time, however, the distribution of the cases is not stationary. With this we mean that the distribution of the features of the process cases change with time.
For example, in a production process of a company selling clothing, the items involved in incoming orders (i.e., cases features) during winter will follow a completely different distribution with respect to items involved in incoming orders during the summer.  Such distribution change may significantly affect the relevance of specific paths in the control-flow of the involved process.

Going back to process model drift, there is a peculiarity of business event streams that cannot be found in traditional data streams. An event log records that a specific activity $a_i$ of a business process $P$ has been executed at time $t$  for a {\it specific} case $c_j$. If the drift from $P$ to $P'$ happens at time $t^*$ {\it while} the process is running, there might be cases for which all the activities have been executed within $P$ (i.e., cases that have terminated their execution before $t^*$), cases for which all the activities have been executed within $P'$ (i.e., cases that have started their execution on or after $t^*$), and cases that have some activities executed within $P$ and some others within $P'$ (i.e., cases that have started their execution before $t^*$ and have terminated after $t^*$). We will refer to these cases as {\it transient cases}. So, under this scenario, the stream will first emit events of cases executed within $P$, followed by events of transient cases, followed by  events of cases executed within $P'$. On the contrary, if the drift {\it does not occur} while the process is running, the stream will first report events referring to complete executions (i.e. cases) of $P$, followed by events referring to complete executions of $P'$ (no transient cases). In any case, the drift is characterized by the fact that $P'$ is very similar to $P$, i.e. the change in the process which emits the events is limited.

Due to space limitation, we restrict our treatment to stationary streams and streams with concept drift with no generation of transient cases. The treatment of other scenarios is left for future work.

%% file: section-approach.tex
\section{Heuristics Miners for Streams} \label{sec:approach}

In this section, we present variants of the Heuristics Miner algorithm (described in \cref{appendix:HM}) to address the SPD problem under different scenarios. First of all, we present two basic algorithms where the standard batch version of Heuristics Miner is used on logs as observation periods extracted from the stream. These algorithms will be used as a baseline reference for the experimental evaluation. Subsequently, a ``fully online'' version of Heuristics Miner, to cope with stationary streams, drift of the process model with no transient cases, and shift of the process model, is introduced.

\subsection{Baseline Algorithm for Stream Mining} \label{sec:basic-approach}

The simplest way to adapt the Heuristics Miner algorithm to deal with streams is to collect events during specific observation periods and then applying the batch version of the algorithm to the current log. This idea is described by \cref{alg:simple_stream} in which two different policies to maintain events in memory are considered.
\begin{algorithm2e}
	\DontPrintSemicolon
	\KwIn{$S$ event stream; $M$ memory of size $\mathit{max}_M$; $P_M$ memory policy (can be `\textit{reset}' or `\textit{shift}')}

	\BlankLine
	\For{\hspace{-0.1cm}\emph{\textbf{ever}}} {
		$e \gets \mathit{observe}(S)$ \tcc*{Observe a new event, where $e = (c_i, a_i,t_i)$}
		\tcc{Check if event $e$ has to be used}
		\If {$\mathit{analyze}(e)$} {

			\BlankLine
			\tcc{Memory update}
			\If{$\mathit{size}(M) = \mathit{max}_M$} {
				\lIf{$P_M$ \textbf{\emph{is}} reset} {$\mathit{reset}(M)$} \;
				\lIf{$P_M$ \textbf{\emph{is}} shift} {$\mathit{shift}(M)$} \;
			}
			$\mathit{insert}(M, e)$ \;

			\BlankLine
			\tcc{Mining update}
			\If{perform mining} {
				$\mathit{HeuristicsMiner}(M)$ \;
			}
		}
	}
	\caption{Sliding Window HM / Periodic Resets HM \label{alg:simple_stream}}
\end{algorithm2e}
Specifically, an event $e$ from the stream $S$ is observed ($e \gets \mathit{observe}(S)$) and analyzed ($\mathit{analyze}(e)$)  to decide if the event has to be considered for mining. If this is the case, it is checked whether there is room in memory to accommodate the event. If the memory is full ($\mathit{size}(M) = \mathit{max}_M$) then the memory policy given as input is adopted. Two different policies are considered: {\it periodic resets}, and {\it sliding windows} \cite[Ch.\ 8]{Aggarwal2007}. In the case of \emph{periodic resets} all the events contained in memory are deleted ($\mathit{reset}$), while in the case of \emph{sliding windows}, only the oldest event is deleted ($\mathit{shift}$). Subsequently, $e$ is inserted in memory and it is checked if it is necessary to perform a mining action. If mining has to be performed, the Heuristics Miner algorithm is executed on the events in memory ($\mathit{HeuristicsMiner}(M)$). Graphical representations of the two policies are reported in \cref{fig:basic-policies}.

\begin{figure}
	\centering
	\subfloat[Periodic reset \label{fig:basic-policies:1}] {
		\includegraphics[width=.5\textwidth]{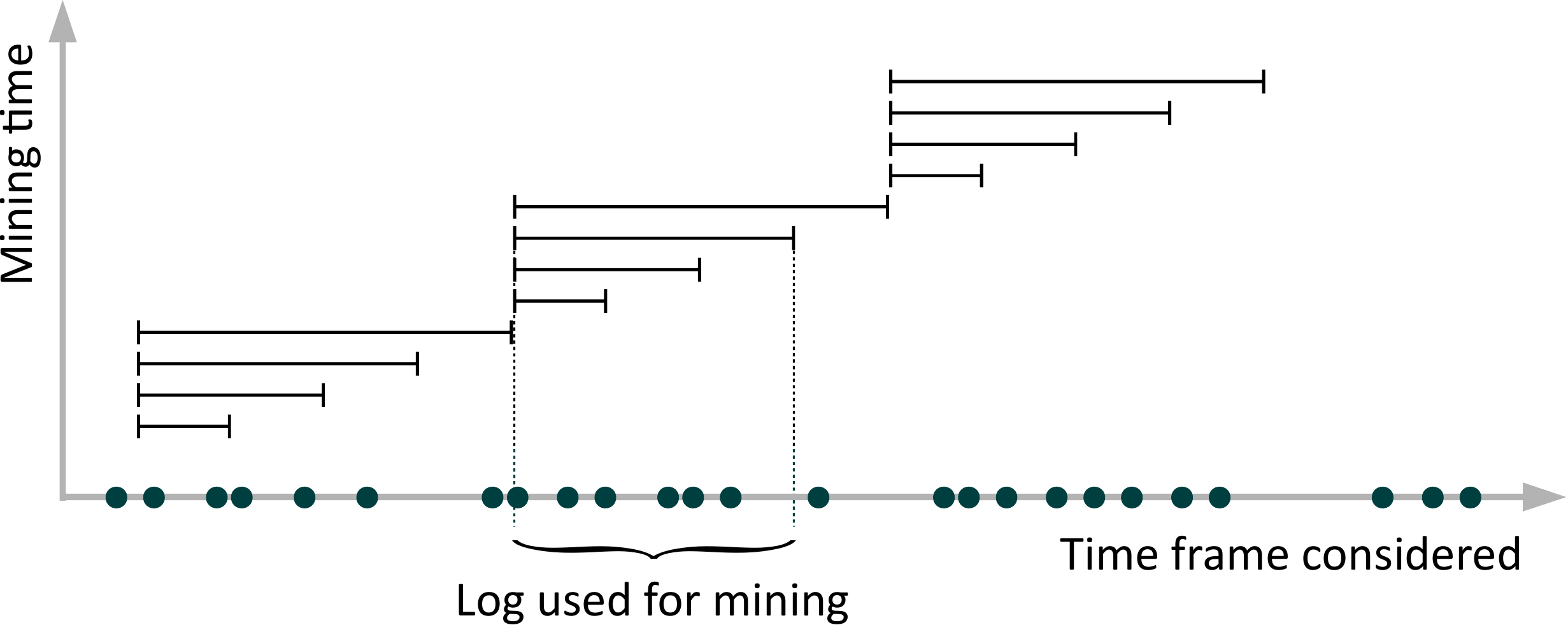}
	}
	\subfloat[Sliding window \label{fig:basic-policies:2}] {
		\includegraphics[width=.5\textwidth]{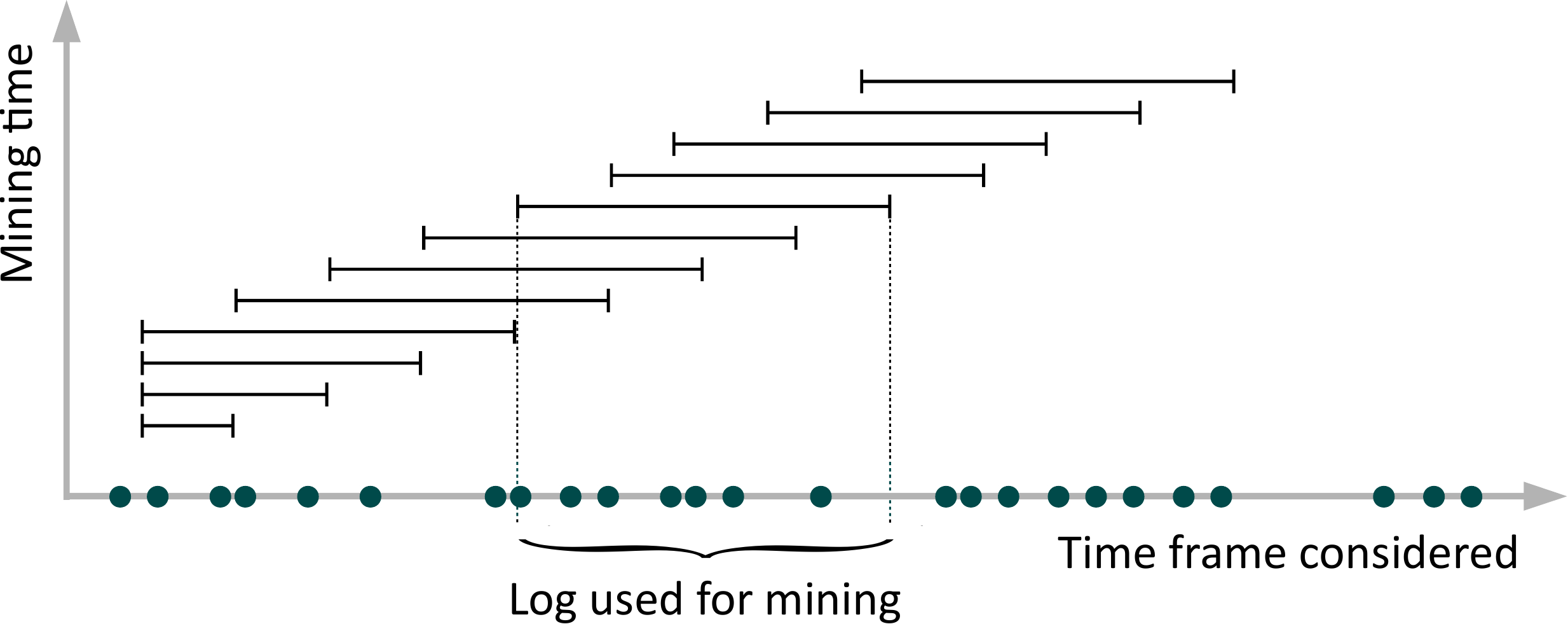}
	}
	\caption{Two basic approaches for the definition of a finite log out of a stream of events. The horizontal segments represent the time frames considered for the mining.}
	\label{fig:basic-policies}
\end{figure}

A potential advantage of the two policies described consists in the possibility to mine the log not only by using Heuristics Miner, but any process mining algorithm (not only for control-flow discovery, for example it is possible to extract information about the social network) already available for traditional batch process discovery techniques. However, the notion of ``history'' is not very accurate: only the more recent events are considered, and an equal importance is assigned to all of them. Moreover, the model is not updated in real-time since each new event received triggers only the update of the log, not necessarily an update of the model: performing a model update for each new event would result in a significant computational burden, well outside the computational limitations assumed for a true online approach. In addition to that, the time required by these approaches is completely unbalanced: when a new event arrives, only inexpensive operations are performed; instead, when the model needs to be
updated, the log retained in memory is mined from scratch. So, every event is handled at least twice: the first time to store it into a log and subsequently any time the mining phase takes place on it. In an online setting, it is more desirable a procedure that does not need to process each event more than once (``one pass algorithm'' \cite{Schweikardt2009}).

\subsection{Stream-Specific Approaches} \label{sec:stream-specific}

In this section, we suggest how to modify the scheme of the basic approaches, so to implement a real online framework, the final approach is described in \cref{online_alg}.
In this framework, the ``current'' log is described in terms of ``latest observed activities'' and ``latest observed dependencies''. Specifically, we define three queues:
\begin{enumerate}
	\item $Q_\mathcal{A}$, with entries in  $\mathcal{A} \times \mathbb{R}$, stores the most recent observed activities jointly with a weight for each activity (that represents its degree of importance with respect to mining);
	\item $Q_\mathcal{C}$, with entries in $\mathcal{C} \times \mathcal{A}$, stores the most recent observed event for each case;
	\item $Q_\mathcal{R}$ with entries in $\mathcal{A} \times \mathcal{A} \times \mathbb{R}$, stores the most recent observed direct succession relations jointly with a weight for each succession relation (that represents its degree of importance with respect to mining).
\end{enumerate}
These queues are used by the online algorithm to retain the information needed to perform mining.

The detailed description of the new algorithm is presented in \cref{online_alg}.
\begin{algorithm2e}
	\DontPrintSemicolon
	\KwIn{$S$ event stream; $\mathit{max}_{Q_\mathcal{A}}, \mathit{max}_{Q_\mathcal{C}}, \mathit{max}_{Q_\mathcal{R}}$ maximum memory sizes for queues $Q_\mathcal{A}$, $Q_\mathcal{C}$, and $Q_\mathcal{R} $, respectively;
	$f_{W_\mathcal{A}}, f_{W_\mathcal{R}}$ model policy; $\mathit{generateModel}(\cdot,\cdot)$.}
	\BlankLine
	\For{\hspace{-0.1cm}\emph{\textbf{ever}}} {
		$e \gets \mathit{observe}(S)$ \tcc*{observe a new event, where $e = (c_i, a_i,t_i)$}
		\tcc{check if event $e$ has to be used}
		\If {$\mathit{analyze}(e)$}{
					\BlankLine
			\uIf {$\not\!\exists (a, w) \in Q_\mathcal{A}$ s.t. $a=a_i$} {
				\If {$\mathit{size}(Q_\mathcal{A}) = \mathit{max}_{Q_\mathcal{A}}$} {
					$\mathit{removeLast}(Q_\mathcal{A})$ \tcc*{removes last entry of $Q_\mathcal{A}$}
				}
				$w \gets 0$ \; \label{algline:setw}
			} \Else {
				$w \gets get(Q_\mathcal{A}, a_i)$ \tcc*{$get$ returns the old weight $w$ of $a_i$ and removes $(a_i,w)$}
			}
			$\mathit{insert}(Q_\mathcal{A}, (a_i, w))$  \tcc*{inserts in front of $Q_\mathcal{A}$}
			$Q_\mathcal{A} \gets f_{W_\mathcal{A}}(Q_\mathcal{A})$ \tcc*{updates the weights of $Q_\mathcal{A}$}
			\BlankLine
			\uIf{$\exists (c,a) \in Q_\mathcal{C}$ s.t. $c=c_i$} {
				$a \gets get(Q_\mathcal{C}, c_i)$  \tcc*{$get$ returns the old activity $a$ of $c_i$ and removes $(c_i,a)$}
				\uIf {$\not\!\exists (a_s, a_f, u) \in Q_\mathcal{R}$ s.t. $(a_s=a )\wedge (a_f = a_i)$} {
					\If{$\mathit{size}(Q_\mathcal{R}) = \mathit{max}_{Q_\mathcal{R}}$} {
						$\mathit{removeLast}(Q_\mathcal{R})\!\!\!\!$ \tcc*{removes last entry of $Q_\mathcal{R}$}
					}
				$u \gets 0$ \; \label{algline:setu}
				} \Else {
					$u \gets \mathit{get}(Q_\mathcal{R}, a, a_i)$ \tcc*{$get$ returns the old weight $u$ of relation $a \to a_i$ and removes $(a,a_i,u)$}
				}
				$\mathit{insert}(Q_\mathcal{R}, (a, a_i,u))$\tcc*{inserts in front of $Q_\mathcal{R}$}
				$Q_\mathcal{R} \gets f_{W_\mathcal{R}}(Q_\mathcal{R})$ \tcc*{updates the weights of $Q_\mathcal{R}$}
			} \ElseIf {$\mathit{size}(Q_\mathcal{C}) = \mathit{max}_{Q_\mathcal{C}}$} {
				$\mathit{removeLast}(Q_\mathcal{C})$ \tcc*{removes last entry of $Q_\mathcal{C}$}
			}
			$\mathit{insert}(Q_\mathcal{C}, (c_i, a_i))$  \tcc*{inserts in front of $Q_\mathcal{C}$}
			\BlankLine
			\tcc{generate model}
			\If{model} {
				$\mathit{generateModel}(Q_\mathcal{A}, Q_\mathcal{R})$ \;
			}
		}
	}

	\caption{Online HM \label{online_alg}}
\end{algorithm2e}
Specifically, the algorithm runs forever, considering, at each round, the current observed event $e = (c_i,a_i,t_i)$. For each current event, it is checked if $a_i$ is already in $Q_\mathcal{A}$. If this is not the case, $a_i$ is inserted in $Q_\mathcal{A}$ with weight $0$. If $a_i$ is already present in the queue, it is removed from its current position and moved at the beginning of the queue. In any case, before insertion, it is checked if $Q_\mathcal{A}$ is full. If this is the case, the oldest stored activity, i.e. the last in the queue, is removed. Subsequently, the weights of $Q_\mathcal{A}$ are updated by $f_{W_A}$. After that, queue $Q_\mathcal{C}$ is examined to look for the most recent event observed for case $c_i$. If a pair $(c_i,a)$ is found, it is removed from the queue, an instance of the succession relation $(a,a_i)$ is created and searched in $Q_\mathcal{R}$. If it is found, it is moved from the current position to the beginning of $Q_\mathcal{R}$. If it is a new succession relation, its
weight is set to $0$.
In any case, before insertion, it is checked if $Q_\mathcal{R}$ is full. If this is the case, the oldest stored relation, i.e. the last in the queue, is removed.  Subsequently, the weights of $Q_\mathcal{R}$ are updated by $f_{W_R}$. Next, after checking if $Q_\mathcal{C}$ is full (in which case the oldest stored event is removed), the event $e$ is stored in $Q_\mathcal{C}$.

Finally, it is checked if a model has to be generated. If this is the case, the procedure $\mathit{generateModel}(Q_\mathcal{A}, Q_\mathcal{R})$ is executed taking as input the current version of queues   $Q_\mathcal{A}$ and $Q_\mathcal{R}$ and producing ``classical'' model representations, such as Causal Nets \cite{VanderAalst2011a} or Petri Nets.

\Cref{online_alg} is parametric with respect to:
\begin{inparaenum}[\itshape i)]
	\item the way weights of queues $Q_\mathcal{A}$ and $Q_\mathcal{R}$ are updated by $f_{W_\mathcal{A}}$, $f_{W_\mathcal{R}}$, respectively;
	\item how a model is generated by $\mathit{generateModel}(Q_\mathcal{A}, Q_\mathcal{R})$.
\end{inparaenum}
In the following, $\mathit{generateModel}(\cdot,\cdot)$ will correspond to the procedure defined by Heuristics Miner (\cref{appendix:HM}). In particular it is possible to consider $Q_\mathcal{A}$ as the counter of activities (to filter out only the most frequent ones) and $Q_\mathcal{R}$ as the counter of direct succession relations, which are used for the computation of the dependency values between pairs of activities.
The following subsections presents some specific instances for $f_{W_\mathcal{A}}$ and $f_{W_\mathcal{R}}$.

\subsubsection{Online Heuristics Miner (Stationary Streams)}

In the case of stationary streams, we can reproduce the behavior of Heuristics Miner as follows. $Q_\mathcal{A}$ should contain, for each activity $a$, the number of occurrences of $a$ observed in $S$ till the current time. Similarly, $Q_\mathcal{R}$ should contain, for each succession $(a,b)$, the number of occurrences of $(a,b)$ observed in $S$ till the current time. Thus both $f_{W_\mathcal{A}}$ and $f_{W_\mathcal{R}}$ must just increment the weight of the first element of the queue:
$$
	f_{W_\mathcal{A}}((a, w)) = \begin{cases}
		(a, w+1) & \text{if } \mathit{first}(Q_\mathcal{A}) = (a, w) \\
		(a, w) & \text{otherwise}
	\end{cases}
$$
$$
	f_{W_\mathcal{R}}((a, b, w)) = \begin{cases}
		(a, b, w+1) & \text{if } \mathit{first}(Q_\mathcal{R}) = (a, b, w) \\
		(a, b, w) & \text{otherwise}
	\end{cases}
$$
where $\mathit{first}(\cdot)$ returns the first element of the queue.

In case of stationary streams, it is possible to use the Hoeffding bound to derive error bounds on the measures computed by the online version of Heuristics Miner. These bounds became tighter and tighter with the increase of the number of processed events. \cref{sec:error-bound} reports some details on that.

It must be noticed that if the sizes of the queues are large enough, the Online Heuristics Miner collects all the needed statistics from the beginning of the
stream till the current time. So it performs very well, provided that the activity distribution of the stream is stationary. However, in real world business processes it is
natural to observe variations both in events distribution and in the workflow of the process generating the stream (concept drift).

In order to cope with concept drift, more importance should be given to more recent events than to older ones. In the following we present a variant of Online Heuristics Miner able to do that.

\subsubsection{Online Heuristics Miner with Aging (Evolving Streams)}

The idea, in this case, is to decrease the weights for the events (and relations) over time when they are not observed. So, every time a new event is observed, only the weight of its activity (and observed succession) is increased, all the others are reduced. Given an ``aging factor'' $\alpha \in [0, 1)$, the weight functions $f_{W_\mathcal{A}}$ (for activities) and $f_{W_\mathcal{R}}$ (for succession relations) are modified so to replace all the occurrences of $w$ on the right side of the equality with $\alpha w$:
$$
	f_{W_\mathcal{A}}((a, w)) = \begin{cases}
		(a, (\alpha w) + 1) & \text{if } \mathit{first}(Q_\mathcal{A}) = (a, w) \\
		(a, \alpha w) & \text{otherwise}
	\end{cases}
$$
$$
	f_{W_\mathcal{R}}((a, b, w)) = \begin{cases}
		(a, b, (\alpha w) + 1) & \text{if } \mathit{first}(Q_\mathcal{R}) = (a, b, w) \\
		(a, b, \alpha w) & \text{otherwise}
	\end{cases}
$$

The basic idea of these new functions is to decrease the ``history'' (i.e., the current number of observations) by an aging factor $\alpha$ (in the formula: $\alpha w$) before increasing it by 1 (the new observation).

These new functions decrease all the weights associated to either an event or a succession relation according to the aging factor $\alpha$ which determines the ``speed'' in forgetting an activity or succession relation, however the most recent observation (the first in the respective queue) is increased by 1. Notice that, if an activity or succession relation is not observed for $t$ time steps, its weight becomes $\alpha^t$. Thus the value of $\alpha$ allows controlling the speed of ``forgetting'': the closer $\alpha$ is to $0$ the faster the weight associated to an activity (or succession relation) that has not been observed for some time goes to $0$, thus to allow the miner to assign larger values to recent events. In this way the miner is more sensitive to sharp variations of the event distribution (concept shift), however the output (generated models) may be less stable because the algorithm becomes more sensitive to random fluctuations of the sampling distribution. When the value of $\alpha$ is close to
$1$, activities that have not been observed recently, but were seen more often some time ago, are able to retain their significance, thus allowing the miner to be able to cope with mild variations of the event distribution (concept drift), but not so reactive in case of concept shift.

One drawback of this approach is that, while it is able to ``forget'' old events, it is not able, at time $t$, to preserve precise statistics for the last $k$ observations and to completely drop observations occurred before time $t-k$. This ability could be useful in case a sudden drastic change in the event distribution.

\subsubsection{Online Heuristics Miner with Self-Adapting Aging (Evolving Stream)}

The third approach explored in this section introduces $\alpha$ as a parameter to control the importance of the ``history'' for the mining: the closer it is to $1$, the more importance is given to the history. The value of $\alpha$, should be decided according to the known degree of ``non-stationarity'' of the stream; however, this information might not be available or it might not be fixed (for example, the process is stationary for a period, then it evolves, and then it becomes stationary again). To handle these cases, it is possible to dynamically adapt the value of $\alpha$. In particular, the idea is \emph{to lower the value of $\alpha$ when a drift is observed and to increase it when the stream seems to be stationary}.

A possible approach to detect the drift is to monitor for variations on the fitness value. This measure, evaluated at a certain period, can be considered as the amount of events (considering only the latest ones) that the current mined process is able to explain. When the fitness value changes drastically, it is likely that a drift has occurred. Using the drift detection, it is possible to adapt $\alpha$ according to the following rules:
\begin{itemize}
	\item if the fitness \emph{decreases} (i.e. there is a drift) $\alpha$ should decreases too (up to $0$), in order to allow the current model to adapt to the new data;
	\item if the fitness \emph{remains unchanged} (i.e. it is within a small interval), it means that there is no drift so the value of $\alpha$ should be increased (up to $1$);
	\item if the fitness \emph{increases}, $\alpha$ should be increased too (up to $1$).
\end{itemize}
The experiments, presented on the next section, consider only variations of $\alpha$ by a constant factor. Alternative update policies (e.g. making the speed of change of $\alpha$ proportional to the observed fitness change) can be considered and is in fact a topic of future investigations.

Early explorations seem to reveal that the effectiveness of the $\alpha$ update policy heavily depends on the problem type (i.e. characteristics of the event of stream), however this topic still requires more investigations.

\subsection{Stream Process Mining with Lossy Counting (Evolving Stream)} \label{sec:lossy}

The approach presented in this section is an adaptation of an existing technique, used for approximate frequency count. In particular, we modified the ``Lossy Counting'' algorithm described in \cite{Manku2002}. We preferred this approach to Sticky Sampling (described in the same paper) since authors stated that, in practice, Lossy Counting performs better. The entire procedure is presented in \cref{alg:lossy-counting}.

\begin{algorithm2e}
	\begin{small}
	\DontPrintSemicolon
	\KwIn{$S$ event stream; $N$ the bucket counter (initially value 1); $\mathcal{D}_A$ activities set; $\mathcal{D}_C$ cases set; $\mathcal{D}_R$ relations set; $\mathit{generateModel}(\cdot,\cdot)$.}
	\BlankLine
	$w \gets \left\lceil \frac{1}{\epsilon} \right\rceil$ \tcc*{define the bucket width}
	\For{\hspace{-0.1cm}\emph{\textbf{ever}}} {
		$b_\mathit{current} = \left\lceil \frac{N}{w} \right\rceil$ \tcc*{define the current bucket id}
		$e \gets \mathit{observe}(S)$ \tcc*{observe a new event, where $e = (c_i, a_i,\Delta_i)$}
		\BlankLine
		\tcc{update the $\mathcal{D}_A$ data structure}
		\eIf{$\exists (a, f, \Delta) \in \mathcal{D}_A$ such that $a = a_i$} {
			Remove the entry $(a, f, \Delta)$ from $\mathcal{D}_A$ \;
			$\mathcal{D}_A \gets (a, f+1, \Delta)$ \tcc*{updates the frequency of element $a_i$}
		}{
			$\mathcal{D}_A \gets \mathcal{D}_A \cup \{(a_i, 1, b_{\textit{current}}-1) \}$ \tcc*{inserts the new observation}
		}
		\tcc{update the $\mathcal{D}_C$ data structure}
		\eIf{$\exists (c, a, f, \Delta) \in \mathcal{D}_C$ such that $c = c_i$} {
			Remove the entry $(c, a, f, \Delta)$ from $\mathcal{D}_C$ \;
			$\mathcal{D}_C \gets (c, a_i, f+1, \Delta)$ \tcc*{updates the frequency and last activity of case $c_i$}
			\tcc{update the $\mathcal{D}_R$ data structure}
			Build relation $r_i$ as $a \to a_i$ \;
			\eIf{$\exists (r, f, \Delta) \in \mathcal{D}_R$ such that $r = r_i$} {
				Remove the entry $(r, f, \Delta)$ from $\mathcal{D}_R$ \;
				$\mathcal{D}_R \gets (r, f+1, \Delta)$ \tcc*{updates the frequency of element $r_i$}
			}{
				$\mathcal{D}_R \gets \mathcal{D}_R \cup \{(r_i, 1, b_{\textit{current}}-1) \}$ \tcc*{adds the new observation}
			}
		}{
			$\mathcal{D}_C \gets \mathcal{D}_C \cup \{(c_i, a_i, 1, b_{\textit{current}}-1) \}$ \tcc*{adds the new observation}
		}
		\tcc{periodic cleanup}
		\If{$N = 0 \mod w$} {
			\ForEach{$(a,f,\Delta) \in \mathcal{D}_A$ such that $f + \Delta \leq b_{\mathit{current}}$} {
				Remove $(a,f,\Delta)$ from $\mathcal{D}_A$ \;
			}
			\ForEach{$(c,a,f,\Delta) \in \mathcal{D}_C$ such that $f + \Delta \leq b_{\mathit{current}}$} {
				Remove $(c,a,f,\Delta)$ from $\mathcal{D}_C$ \;
			}
			\ForEach{$(r,f,\Delta) \in \mathcal{D}_R$ such that $f + \Delta \leq b_{\mathit{current}}$} {
				Remove $(r,f,\Delta)$ from $\mathcal{D}_R$ \;
			}
		}
		$N \gets N +1$ \tcc*{increments the bucket counter}
		\tcc{generate model}
		\If{model} {
			$\mathit{generateModel}(\mathcal{D}_A, \mathcal{D}_R)$ \;
		}
 	}
 	\end{small}
	\caption{Lossy Counting HM \label{alg:lossy-counting}}
\end{algorithm2e}

The basic idea of Lossy Counting algorithm is to conceptually divide the stream into buckets of width $w = \left\lceil \frac{1}{\epsilon} \right\rceil$, where $\epsilon \in (0, 1)$ is an error parameter.
The \emph{current} bucket (i.e., the bucket of the last element seen) is identified with $b_\mathit{current} = \left\lceil \frac{N}{w} \right\rceil$, where $N$ is the progressive events counter.

The basic data structure used by Lossy Counting is a set of entries of the form $(e, f, \Delta)$ where: $e$ is an element of the stream; $f$ is the estimated frequency of the item $e$; and $\Delta$ is the maximum possible error.
Every time a new element $e$ is observed, the algorithm looks whether the data structure contains an entry for the corresponding element. If such entry exists then its frequency value $f$ is incremented by one, otherwise a new tuple is added: $(e, 1, b_\mathit{current}-1)$.
Every time $N \equiv 0 \mod w$, the algorithm cleans the data structure by removing the entries that satisfy the following inequality: $f + \Delta \leq b_\mathit{current}$. Such condition ensures that, every time the cleanup procedure is executed, $b_\mathit{current} \leq \epsilon N$.

This algorithm has been adapted to the SPD problem, using three instances of the basic data structure. In particular, it counts the frequencies of the activities (with the data structure $\mathcal{D}_A$) and the frequencies of the direct succession relations (with the data structure $\mathcal{D}_R$). In order to obtain the relations, a third instance of the same data structure is used, $\mathcal{D}_C$.
In $\mathcal{D}_C$, each item is of the type $(c, a, f, \Delta)$ where $c \in \mathcal{C}$ represent the case identifier; $f$ and $\Delta$, as in previous cases, respectively correspond to the frequency and to the bucket id; and $a \in A$ is the latest activity observed on the corresponding case.
Every time a new activity is observed, $\mathcal{D}_A$ is updated. After that, the procedure checks if, given the case identifiers of the current event, there is an entry in $\mathcal{D}_C$. If this is not the case a new entry is added to $\mathcal{D}_C$ (by adding the current case id and the activity observed). Otherwise, the $f$ and $a$ components of the entry in $\mathcal{D}_C$ are updated.

The Heuristics Miner can be used to generate the model, since a set of dependencies between activities is available.

%% file: section-implementation.tex
\section{Implementation} \label{sec:implementation}

All the approaches presented into this paper have been implemented in the ProM 6.1 toolkit \cite{Verbeek2010}. Moreover, a ``stream simulator'' and a ``logs merger'' have also been implemented to allow for experimentation (to test new algorithms and to compose logs).

Communications between stream sources and stream miner are performed over the network: each event emitted consists of a ``small log'' (i.e., a trace which contains exactly one event), encoded as a XES string \cite{Gunther2009}. An example of an event log streamed is presented in \cref{lst:openxes}. This approach is useful to simulate ``many-to-many environments'' where one source emits events to many miners and one miner can use many stream sources. The current implementation supports only the first scenario (currently it is not possible to mine streams generated by more than one source).

\lstsetOpenXES
\begin{lstlisting}[float,
	language=OpenXES,
	caption={OpenXES fragment streamed over the network.},
	label=lst:openxes]
<log openxes.version="1.0RC7" xes.features="nested-attributes" xes.version="1.0" xmlns="http://www.xes-standard.org/">
	<trace>
		<string key="concept:name" value="case_id_0" />
		<event>
			<date key="time:timestamp" value="2012-04-23T10:33:04.004+02:00" />
			<string key="concept:name" value="A" />
			<string key="lifecycle:transition" value="Task_Execution" />
		</event>
	</trace>
</log>
\end{lstlisting}

\begin{figure}
	\centering
	\includegraphics[width=.8\textwidth]{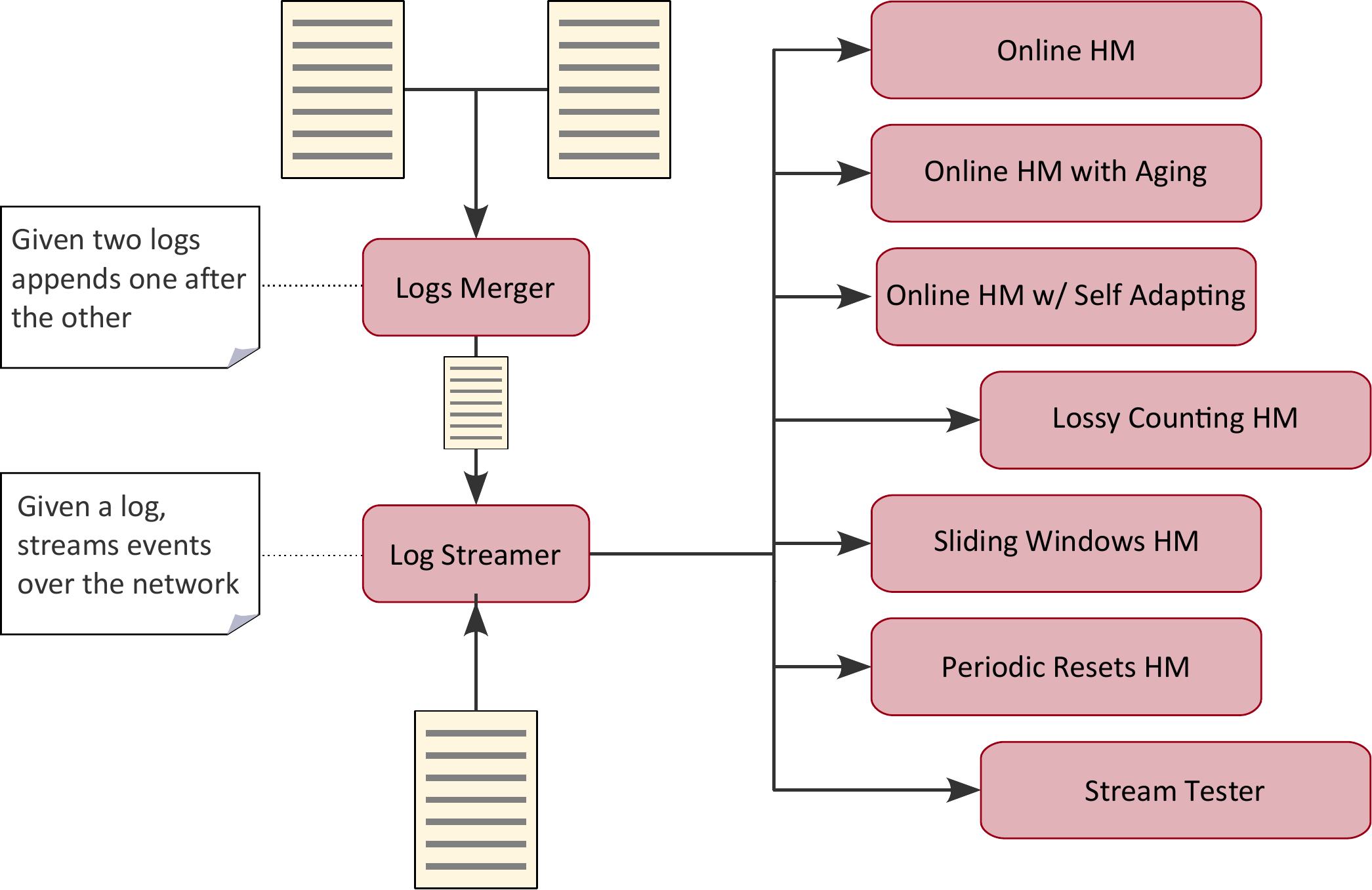}
	\caption{Architecture of the plugins implemented in ProM and how they interact with each other. Each rounded box represents a ProM plugin.}
	\label{fig:plugins-architecture}
\end{figure}

\cref{fig:plugins-architecture} proposes the the set of ProM plugins implemented, and how they interact each other. The available plugins can be split into two groups: plugins for the simulation of the stream and plugins to mine streaming event data. To simulate a stream there is the ``Log Streamer'' plugin. This plugin, receives a static log file as input and streams each event over the network, according to its timestamp (in this context, timestamps are used only to determine the order of events). It is possible to define the time between each event, in order to test the miner under different emission rates (i.e. to simulate different traffic conditions). A second plugin, called ``Logs Merger'' can be used to concatenate different log files generated by different process models, just for testing purposes.

Once the stream is active (i.e. events are sent through the network), the clients can use these data to mine the model. There is a ``Stream Tester'' plugin, which just shows the events received. The other 6 plugins support the two basic approaches (\cref{sec:basic-approach}), and the four stream specific approaches (\cref{sec:lossy} and \ref{sec:stream-specific}).

In a typical session of testing a new stream process mining algorithm, we expect to have two separate ProM instances active at the same time: the first is streaming events over the network and the second is collecting and mining them.

\begin{figure}[t]
	\centering
	\includegraphics[height=3.9cm]{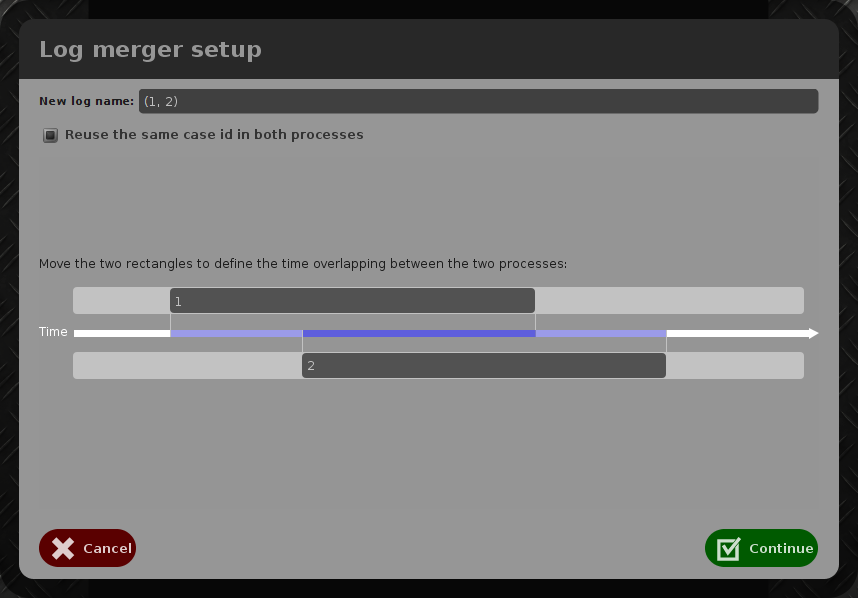} \hfill
	\includegraphics[height=3.9cm]{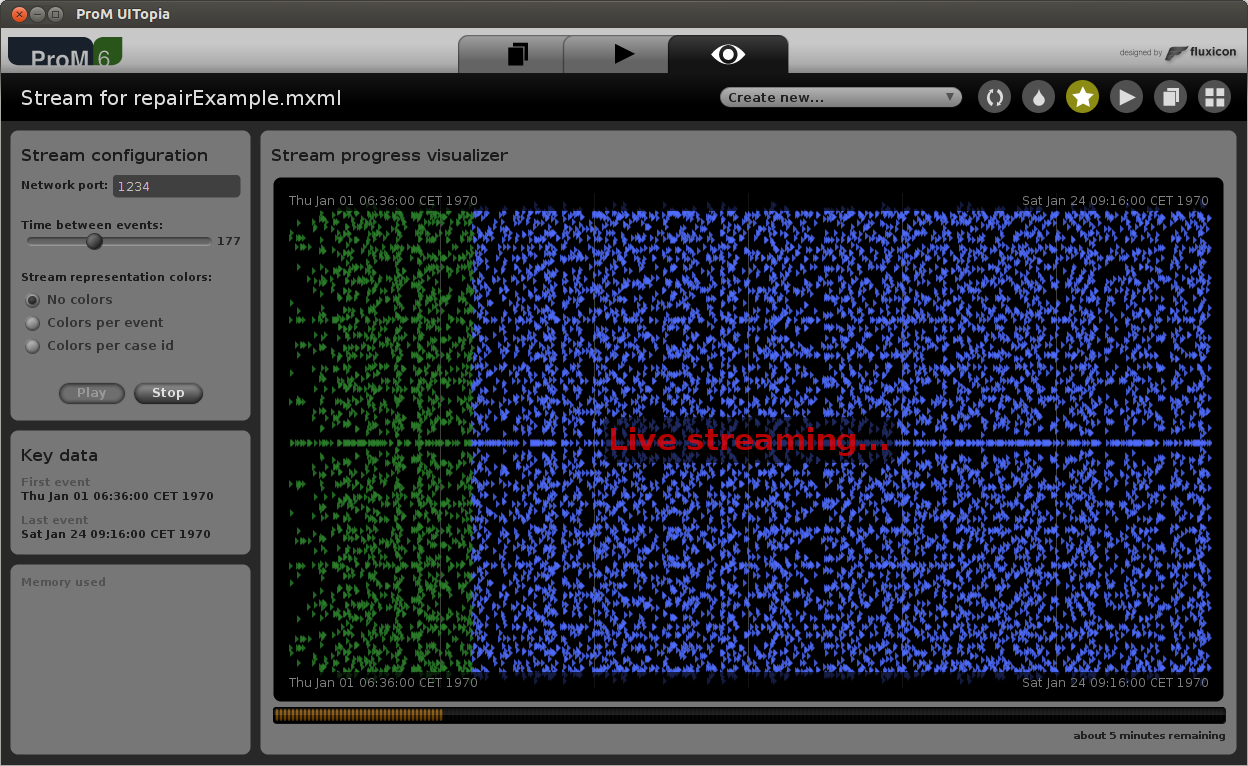} \\[.5em]
	\includegraphics[width=.49\textwidth]{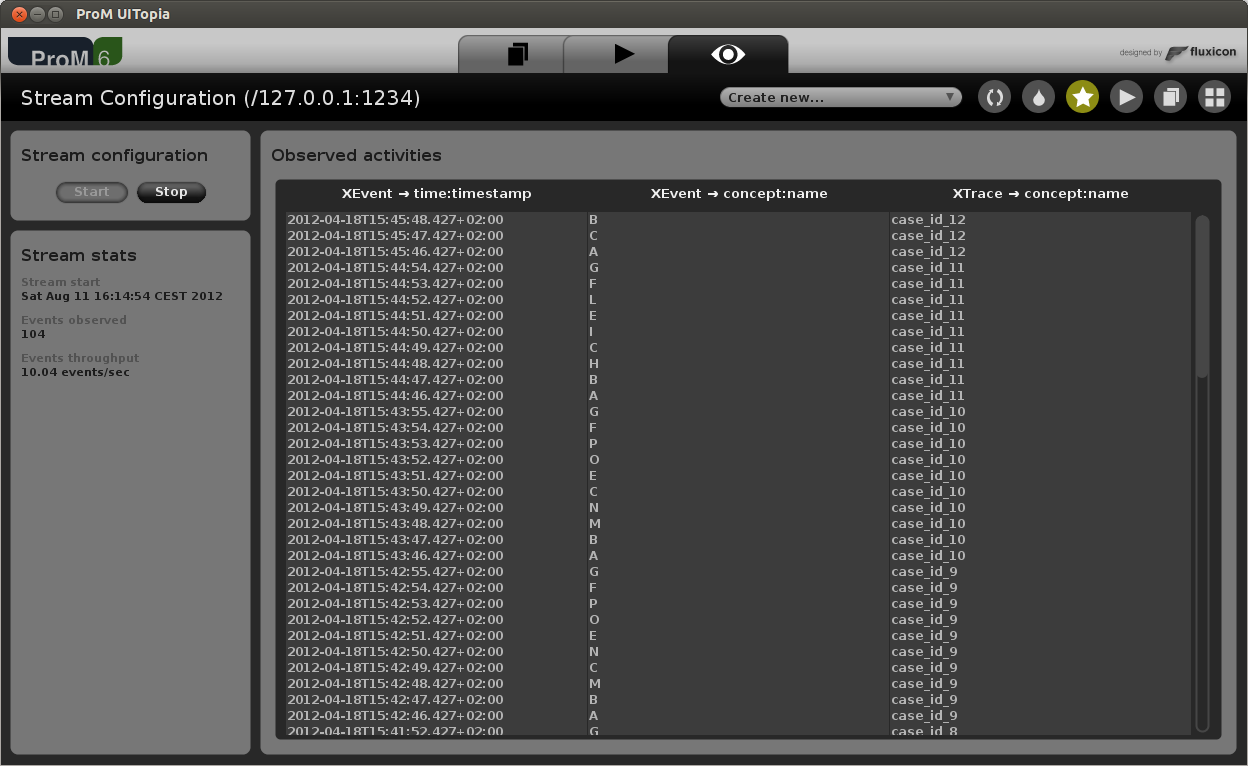} \hfill
	\includegraphics[width=.49\textwidth]{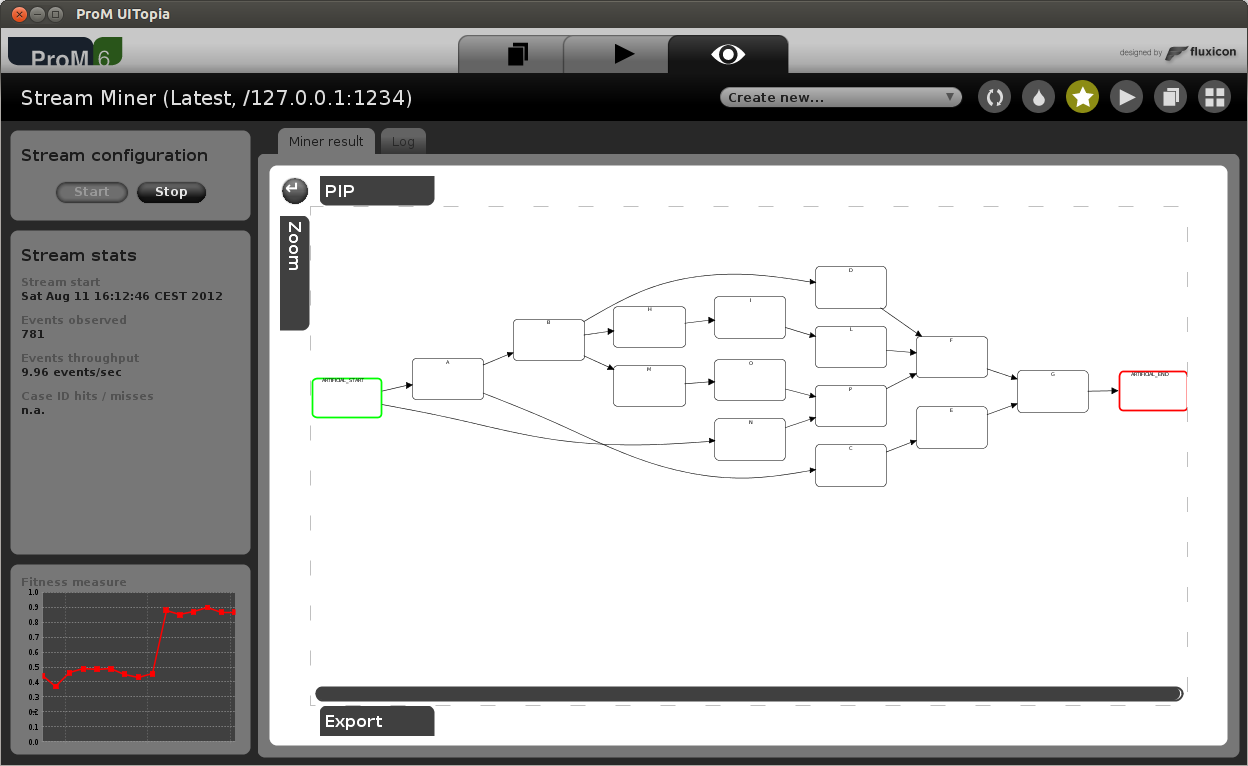}
	\caption{Screenshots of four implemented ProM plugins. The first image (\emph{top left}) shows the logs merger (it is possible to define the overlap level of the two logs); the second image (\emph{top right}) represents the log streamer, the bottom left image is the stream tester and the image at the bottom right shows the Online HM.}
	\label{fig:screenshots}
\end{figure}

\cref{fig:screenshots} contains three screenshots of the ProM plugins implemented.
The first image, on top, contains the process streamer: the left bar describes the stream configuration options (such as the speed or the network port for new connections), the central part contains a representation of the log as a dotted chart \cite{Song2007} (the $x$ axis represents the time, and each point with the same timestamp $x$ value is an event occurred at the same instant). Blue dots are the events that are not yet sent (future events), green ones are the events already streamed (past events). It is possible to change the color of the future events so that every event referring to the same activity or to the same process instance has the same color.
The figure in the middle contains the Stream Tester: each event of a stream is appended to this list, which shows the timestamp of the activity, its name and its case id. The left bar contains some basic statistics (i.e. beginning of the streaming session, number of events observed and average number of events observed per second).
The last picture, at the bottom, represents the Online HM miner. This view can be divided into three parts: the central part, where the process representation is shown (in this case, as a Causal Net); the left bar contains, on top, buttons to start/stop the miner plus some basic statistics (i.e., beginning of the streaming session, number of events observed and average number of events observed per second); at the bottom, there is a graph which shows the evolution of the fitness measure.

Moreover, Command-Line Interface (CLI) versions of the miners are available too\footnote{See \url{http://www.processmining.it} for more details.}. In these cases, events are read from a static file (one event per line) and the miners update the model (this implementation realizes an incremental approach of the algorithm). These implementations are can be run in batch and are used for automated experimentation.

%% file: section-results.tex
\section{Results} \label{sec:results}

The algorithms presented in this paper have been tested using four datasets: event logs from two artificial processes (one stationary and one evolving); a synthetic example; and a real event log.

\subsection{Models description}

\begin{figure}
	\centering
	\includegraphics[width=\textwidth]{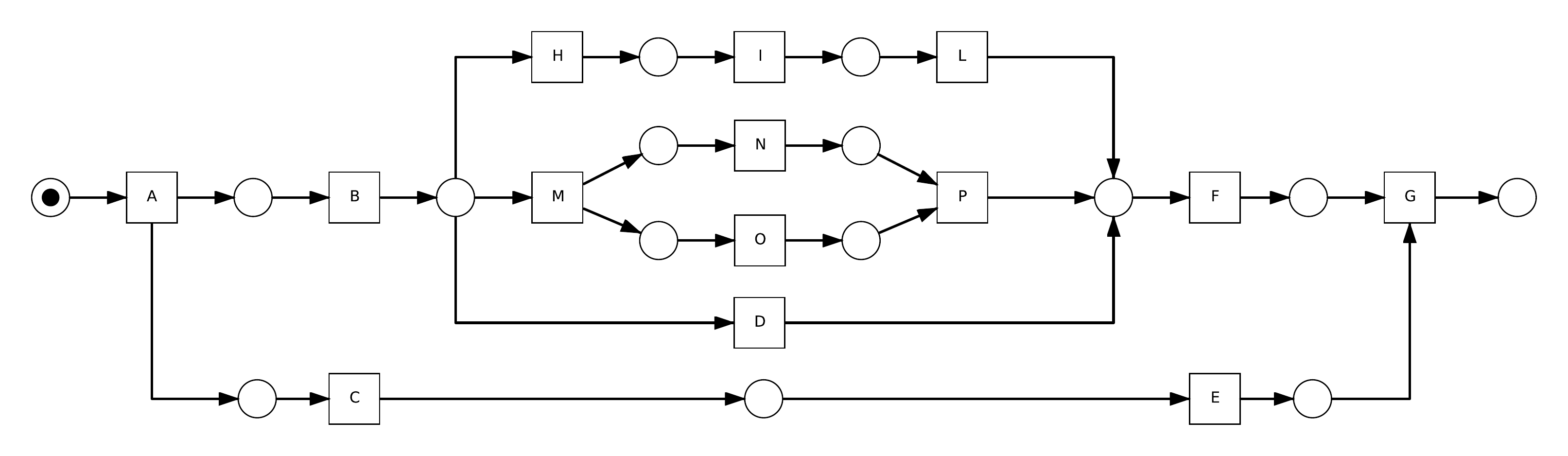}
	\caption{Model 1. Process model used to generate the stationary stream.}
	\label{fig:model:4}
\end{figure}
\begin{figure}
	\includegraphics[width=\textwidth]{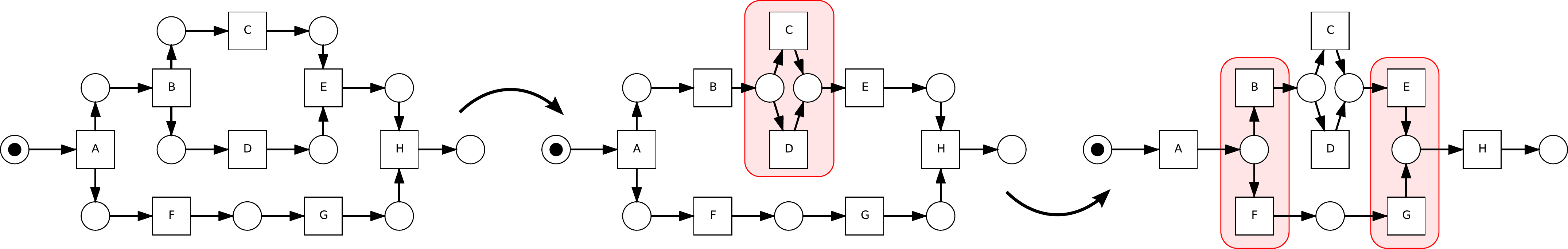}
	\caption{Model 2. The three process models that generate the evolving stream. Red rounded rectangles indicate areas subject to modification.}
	\label{fig:model:6}
\end{figure}

The two artificial processes are shown in \cref{fig:model:4} and \cref{fig:model:6}, both are described in terms of a as Petri Net. The first one describes the complete model (Model 1) that is simulated to generate the stationary stream. The second one (Model 2) presents the three models which are used to generate three logs describing an evolving stream. In this case, the final stream is generated considering the hard shift of the three logs generated from the single process executions.

\begin{figure}
	\includegraphics[width=\textwidth]{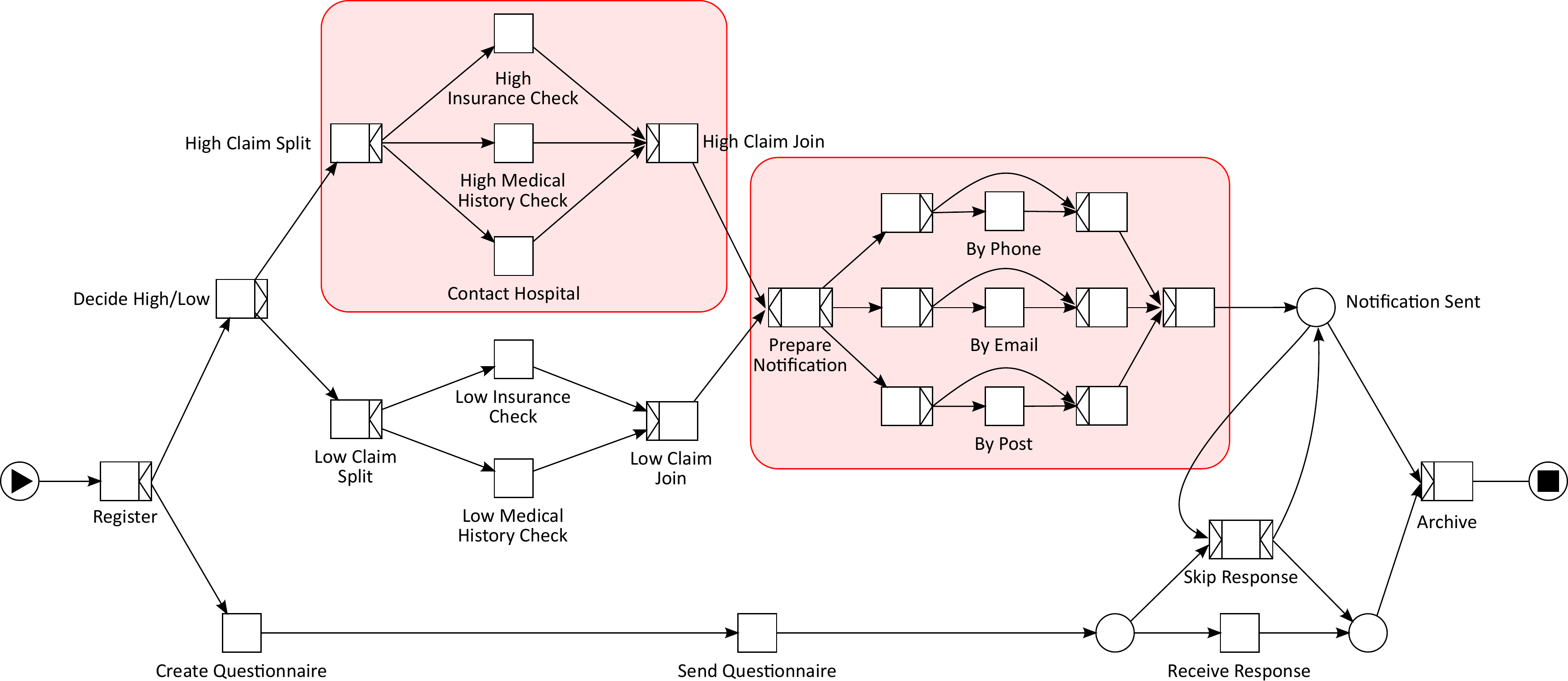}
	\caption{Model 3. The first variant of the third model. Red rounded rectangles indicate areas that will be subject to the modifications.}
	\label{fig:model:3}
\end{figure}

The synthetic example (Model 3) is reported in \cref{fig:model:3}. This example is taken from \cite[Chap. 5]{Bose2012} and is expressed as a YAWL \cite{Vanderaalst2005} process. This model describes a possible health insurance claim process of a travel agency. This example is modified 4 times so, at the end, the stream contains traces from 5 different processes. Also in this case the type of drift is shift. Due to space limitation, only the first process is presented and the red rectangles indicate areas that are modified over time.

\subsection{Algorithms evaluation}

The streams generated from the described models are used for the evaluation of the presented techniques.
There are various metrics to evaluate the process models with respect to an event log. Typically four quality dimensions are considered for comparing model and log:
\begin{inparaenum}[\itshape (a)]
	\item fitness;
	\item simplicity;
	\item precision; and
	\item generalization \cite{VanderAalst2011,VanderAalst2012}.
\end{inparaenum}
In order to measure how well the model describes the log without allowing the reply of traces not generated by the target process, here we measure the performance both in terms of \emph{fitness} (computed according to \cite{Adriansyah2011a}) and in terms of \emph{precision} (computed according to \cite{Munoz-Gama2010}).
The first measure reaches its maximum when all the traces in the log are properly replied by the model, while the second one prefers models that describe a ``minimal behavior'' with respect to all the models that can be generated starting from the same log.
In all experiments, the \emph{fitness} and \emph{precision} measures are computed over the last $x$ observed events (where $x$ varies according to log size), $q$ refers to the maximum size of queues,  and default parameters of Heuristics Miner, for model generation, are used.

The main characteristics of the three streams are:
\begin{itemize}
	\item \emph{Streams for Model 1}: 3448 events, describing 400 cases;
	\item \emph{Streams for Model 2}: 4875 events, describing 750 cases (250 cases and 2000 events for the first process model, 250 cases and 1750 events for the second, and 250 cases with 1125 events for the third one);
	\item \emph{Stream for Model 3}: 58783 events, describing 6000 cases (1199 cases and 11838 events for the first variant; 1243 cases and 11690 events for the second variant; 1176 cases and 12157 events for the third variant; 1183 cases and 10473 events for the fourth variant; and 1199 cases and 12625 events for the fifth variant).
\end{itemize}
We compare the basic approaches versus the different online versions of stream miner, against the different streams.


\begin{figure}[h]
	\centering
	\includegraphics[width=.48\textwidth]{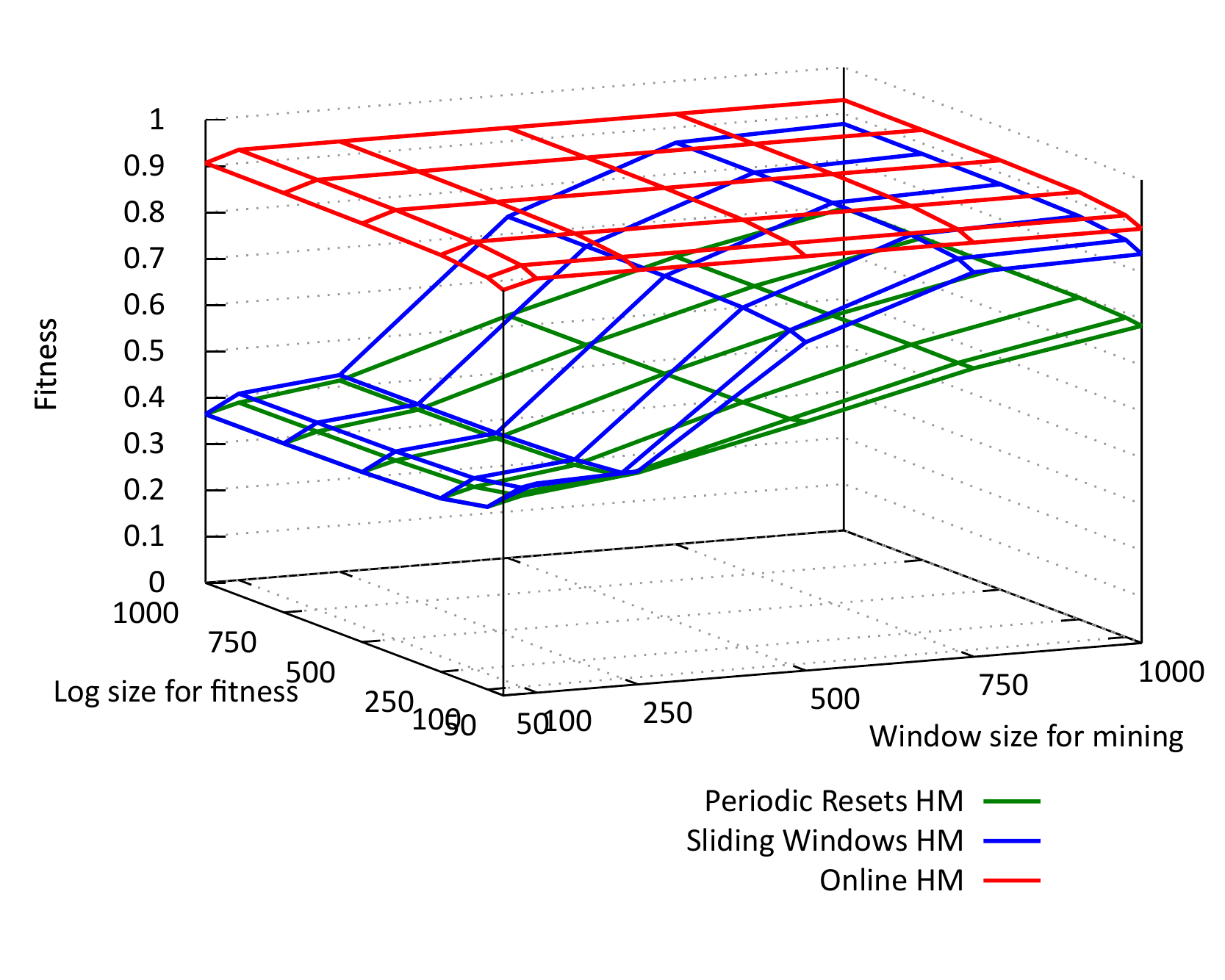}\hspace{.1\textwidth}
	\includegraphics[width=.35\textwidth]{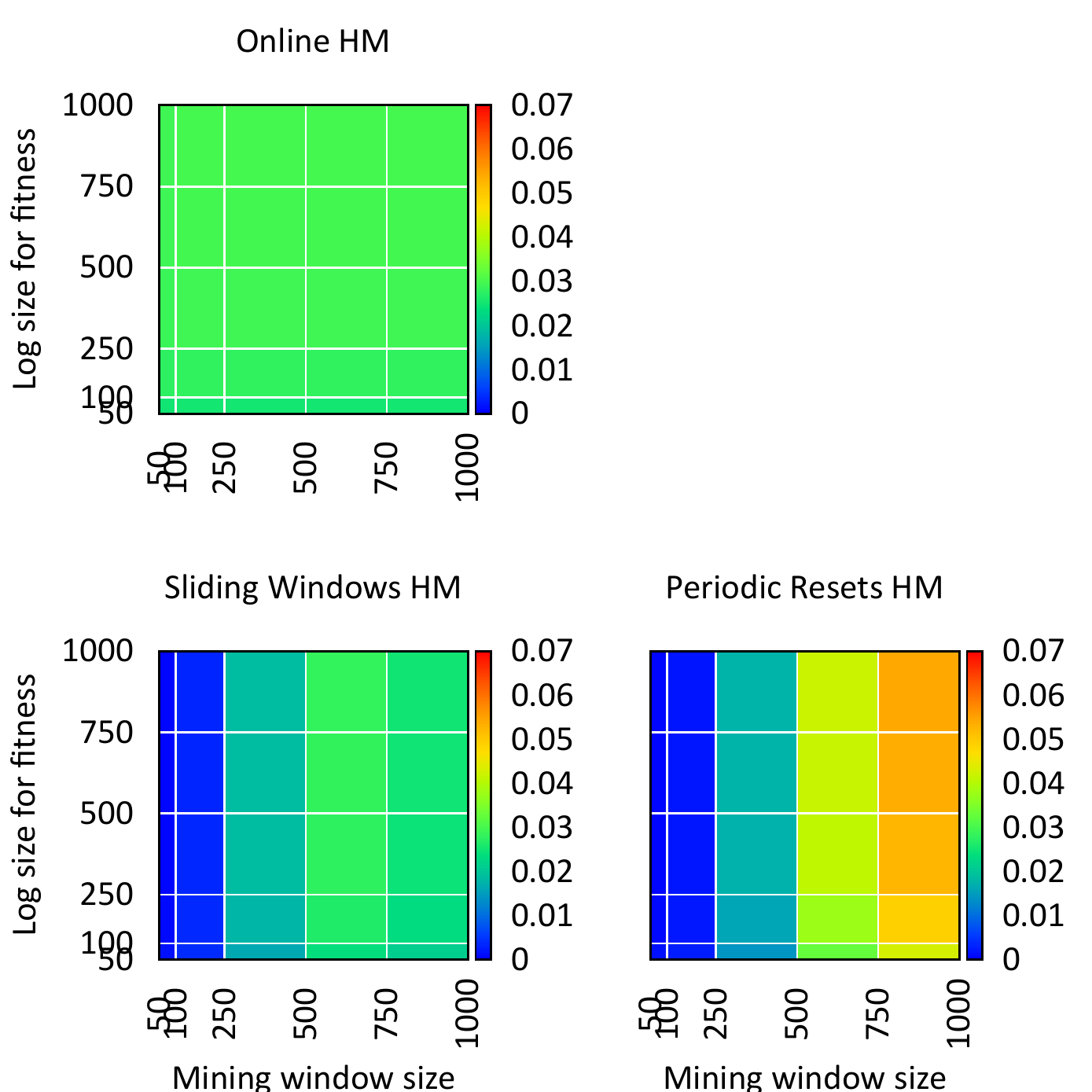}\\[1em]
	\includegraphics[width=.48\textwidth]{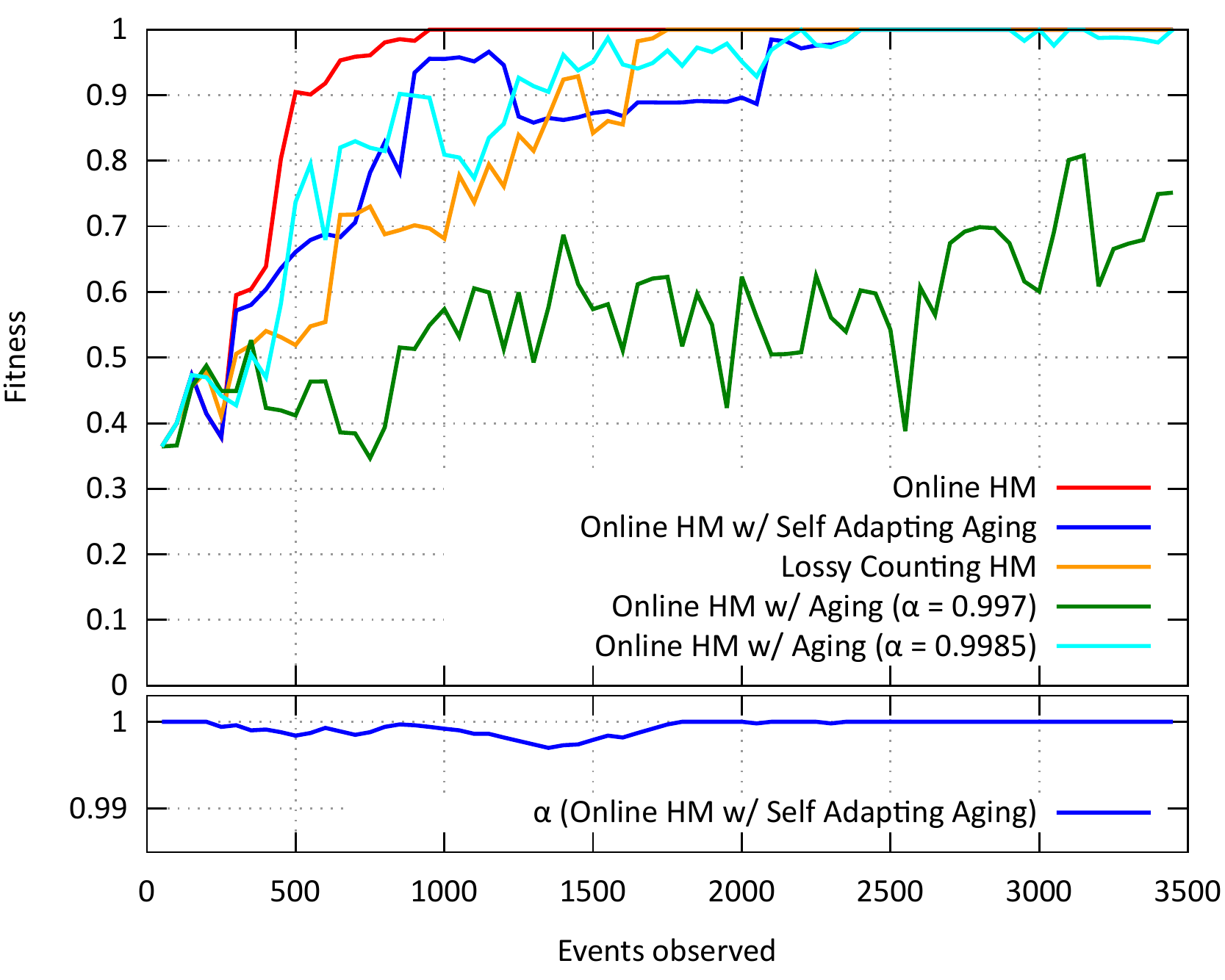}\hfill
	\includegraphics[width=.48\textwidth]{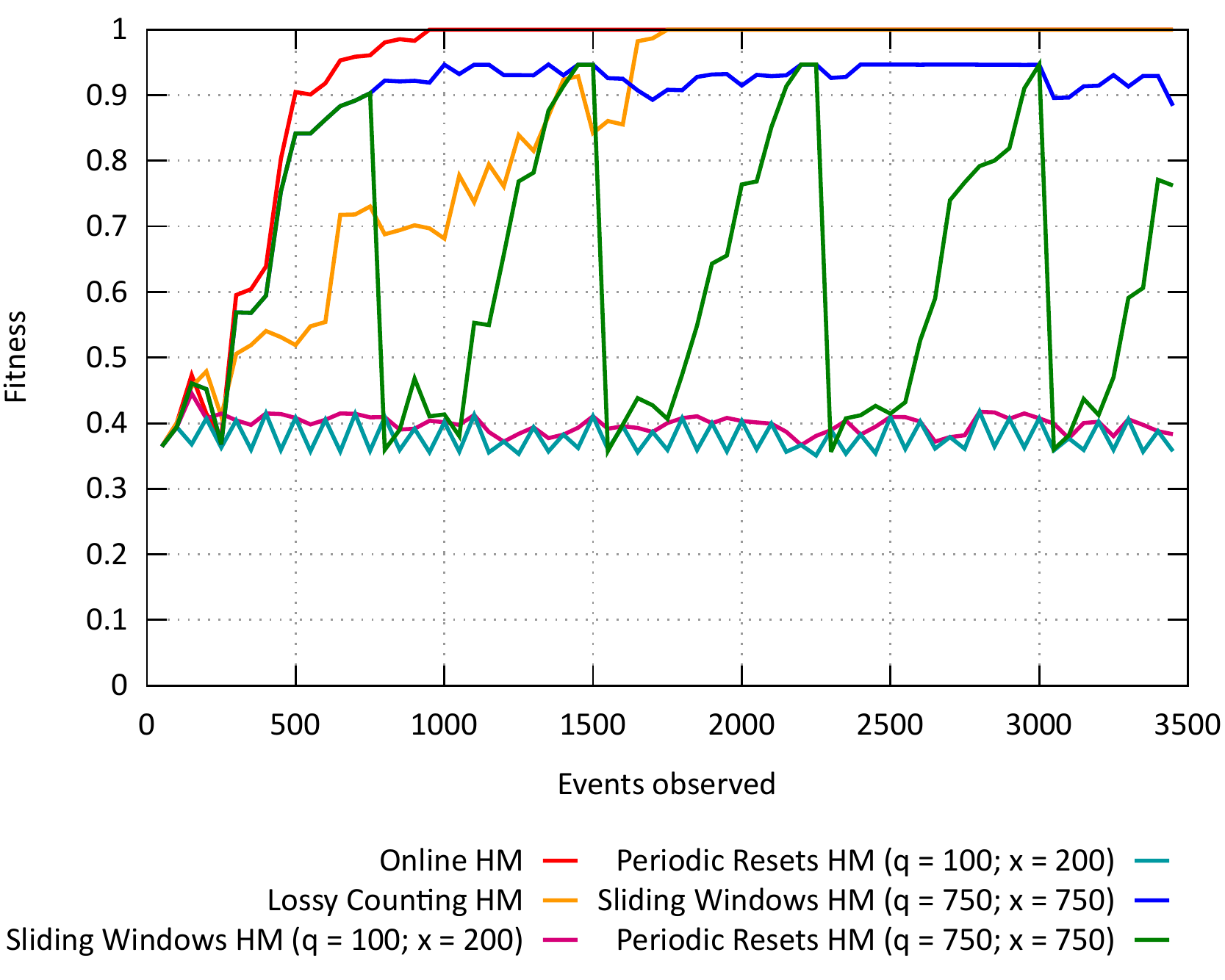}
	\caption{Aggregated experimental results for five streams generated by Model 1. \emph{Top:} average (\emph{left}) and variance (\emph{right}) values of fitness measures for basic approaches and the Online HM. \emph{Bottom:} evolution in time of average fitness for Online HM with queues size 100 and log size for fitness 200; curves for HM with Aging ($\alpha = 0.9985$ and $\alpha = 0.997$), HM with Self Adapting (evolution of the $\alpha$ value is shown at the bottom), Lossy Counting and different configurations of the basic approaches are reported as well.}
	\label{fig:fitness-performance:model1}
\end{figure}

\cref{fig:fitness-performance:model1} reports the aggregated experimental results for five streams generated by Model 1. The two charts on top report the averages of the fitness (left) and the variance (right) for the two basic approaches and the Online HM. The presented values are calculated varying the size of the window used to perform the mining (in the case of Online HM it's the size of the queues), and the number of events used to calculate the fitness measure (i.e. only the latest $x$ events are supposed to fit the model).
For each combination (number of events for the mining and number of events for fitness computation) a run of the miner has been executed (for each of the five streams) and the average and variance values of the fitness (which is calculated every 50 events observed) are reported. It is clear, from the plot, that the Online HM is outperforming the basic approaches, both in terms of speed in finding a good model and in terms of fitness of the model itself.
The bottom of the figure presents, on the left hand side, a comparison of the evolution of the average fitness of the Online HM, the HM with Aging ($\alpha = 0.9985$ and $\alpha = 0.997$), the HM with Self Adapting approach and Lossy Counting. For these runs a queues size of $100$ have been used and, for the fitness computation, the latest $200$ events are considered. In this case, the lossy counting considers an error value $\epsilon = 0.01$. The right hand side of \cref{fig:fitness-performance:model1} compares the basic approaches, with different window and fitness sizes against the Online HM and the Lossy Counting approach.
As expected, since there is no drift, the Online HM outperforms the versions with aging. In fact, HM with aging beside being less stable, degrades performances as the value of $\alpha$ decreases, i.e. less importance is given to less recent events. This is consistent with the bad performance reported for the basic approaches which can exploit only the most recent events contained in the window.
The self adapting strategy, after an initial variation of the $\alpha$ parameter, is able to converge to the Online HM by eventually choosing a value of $\alpha$ equals to 1.


\begin{figure}[h]
	\centering
	\includegraphics[width=.48\textwidth]{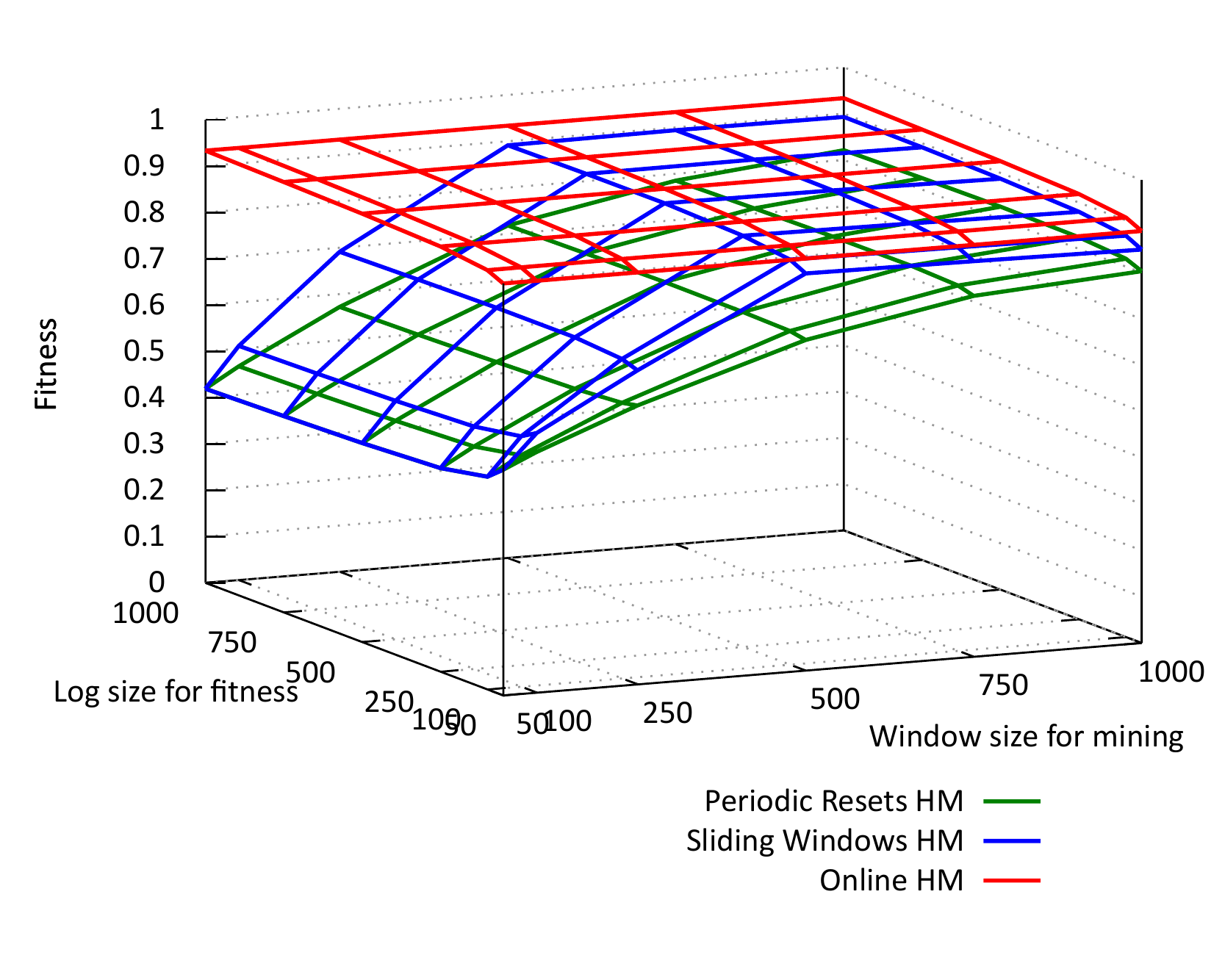}\hspace{.1\textwidth}
	\includegraphics[width=.35\textwidth]{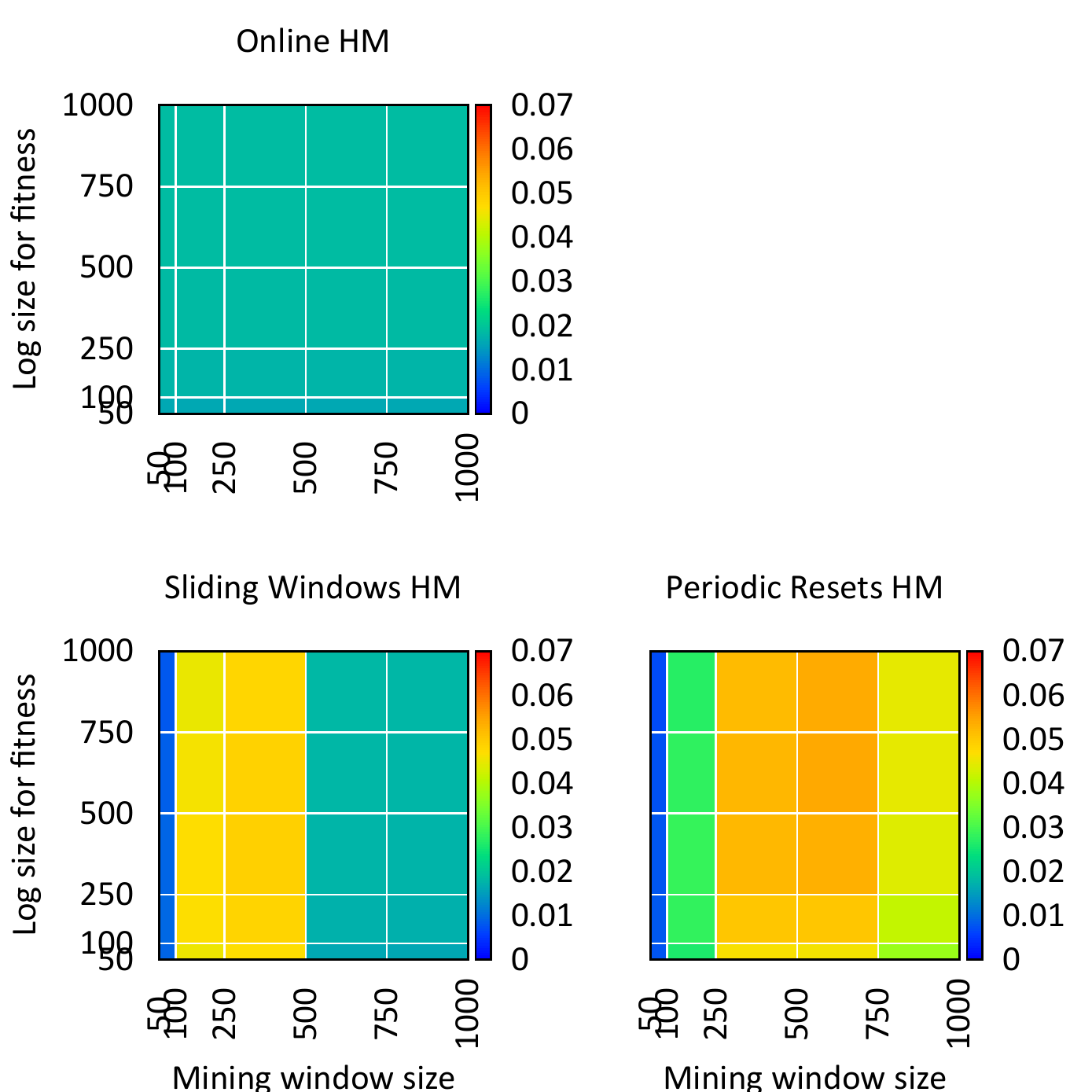}\\[1em]
	\includegraphics[width=.48\textwidth]{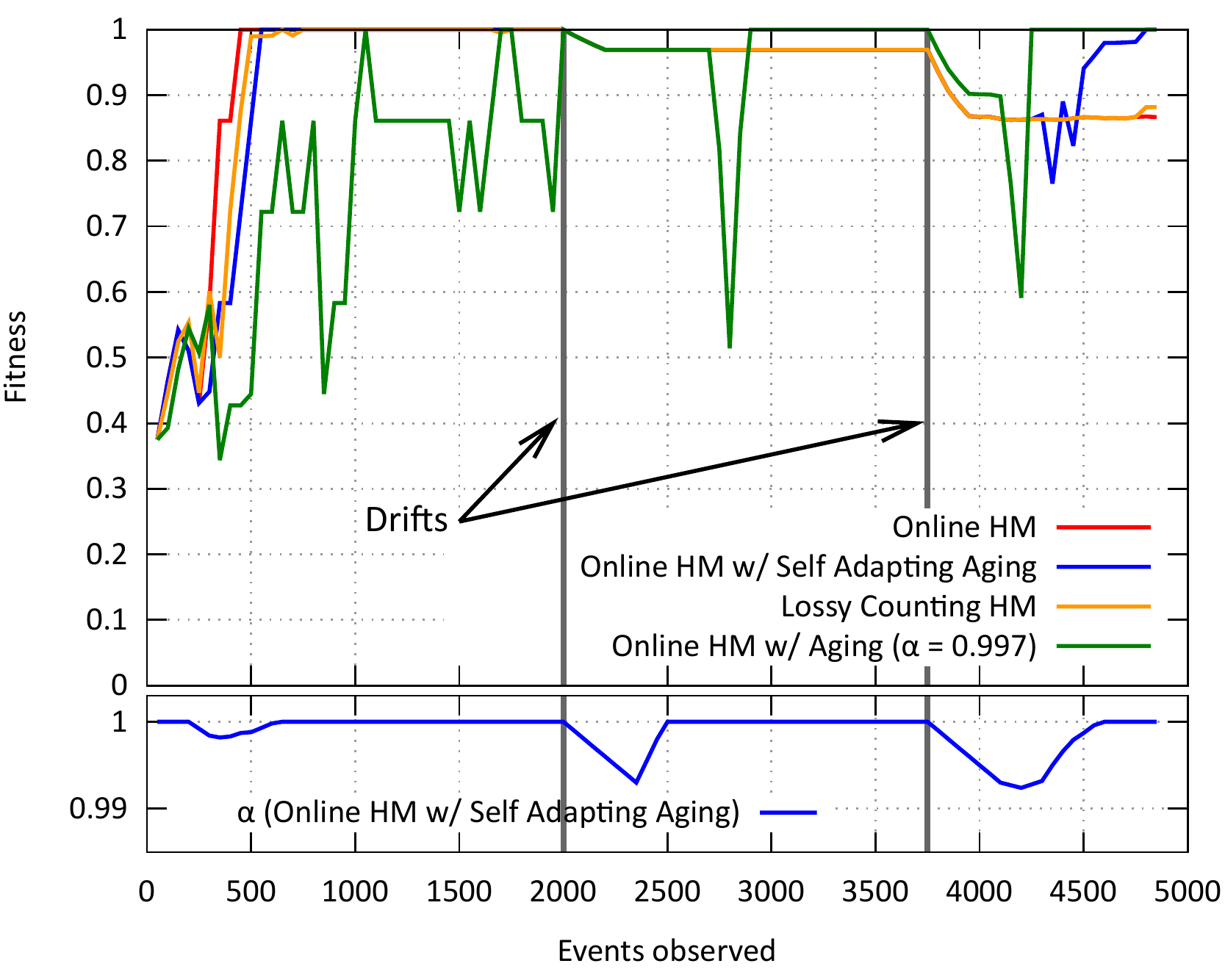}\hfill
	\includegraphics[width=.48\textwidth]{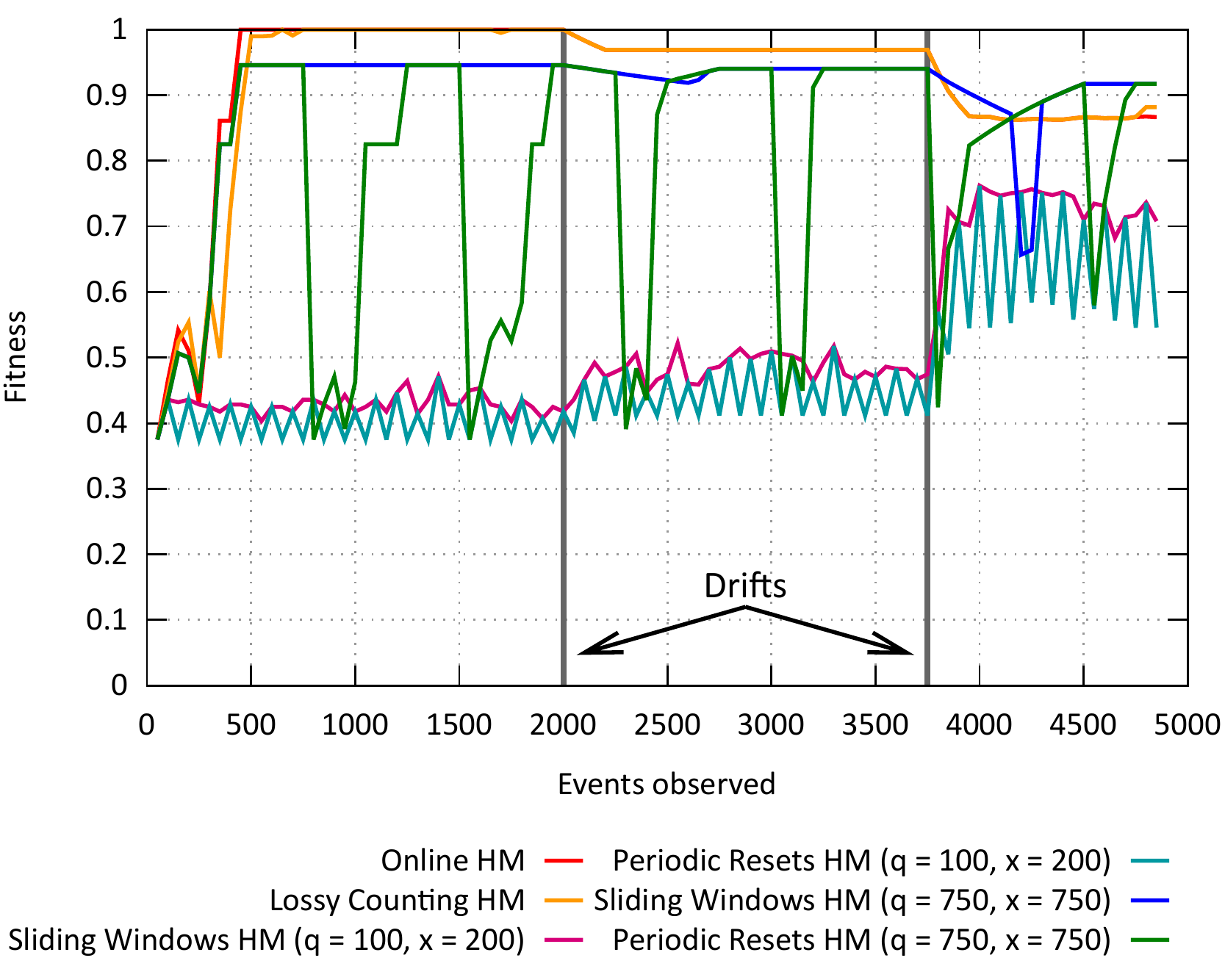}
	\caption{Aggregated experimental results for five streams generated by evolving Model 2. \emph{Top:} average (\emph{left}) and variance (\emph{right}) values of fitness measures for basic approaches and Online HM. \emph{Bottom:} evolution in time of average fitness for Online HM with queues size 100 and log size for fitness 200; curves for HM with Aging ($\alpha = 0.997$), HM with Self Adapting (evolution of the $\alpha$ value is shown at the bottom), Lossy Counting and different configurations of the basic approaches are reported as well. Drift occurrences are marked with vertical bars.}
	\label{fig:fitness-performance:model2}
\end{figure}

\cref{fig:fitness-performance:model2} reports the aggregated experimental results for five streams generated by Model 2. In this case we adopted exactly the same experimental setup, procedure and results presentation as described before. In addition, the occurrences of drift are marked.
As expected, the performance of Online HM decreases at each drift, while HM with Aging is able to recover from the drifts. The price paid for this ability is a less stable behavior. HM with Self Adapting aging seems to be the right compromise being eventually able to recover from the drifts while showing a stable behavior. The $\alpha$ curve shows that the self adapting strategy seems to be able to detect the concept drifts.


\begin{figure}[h]
	\centering
	\includegraphics[width=\textwidth]{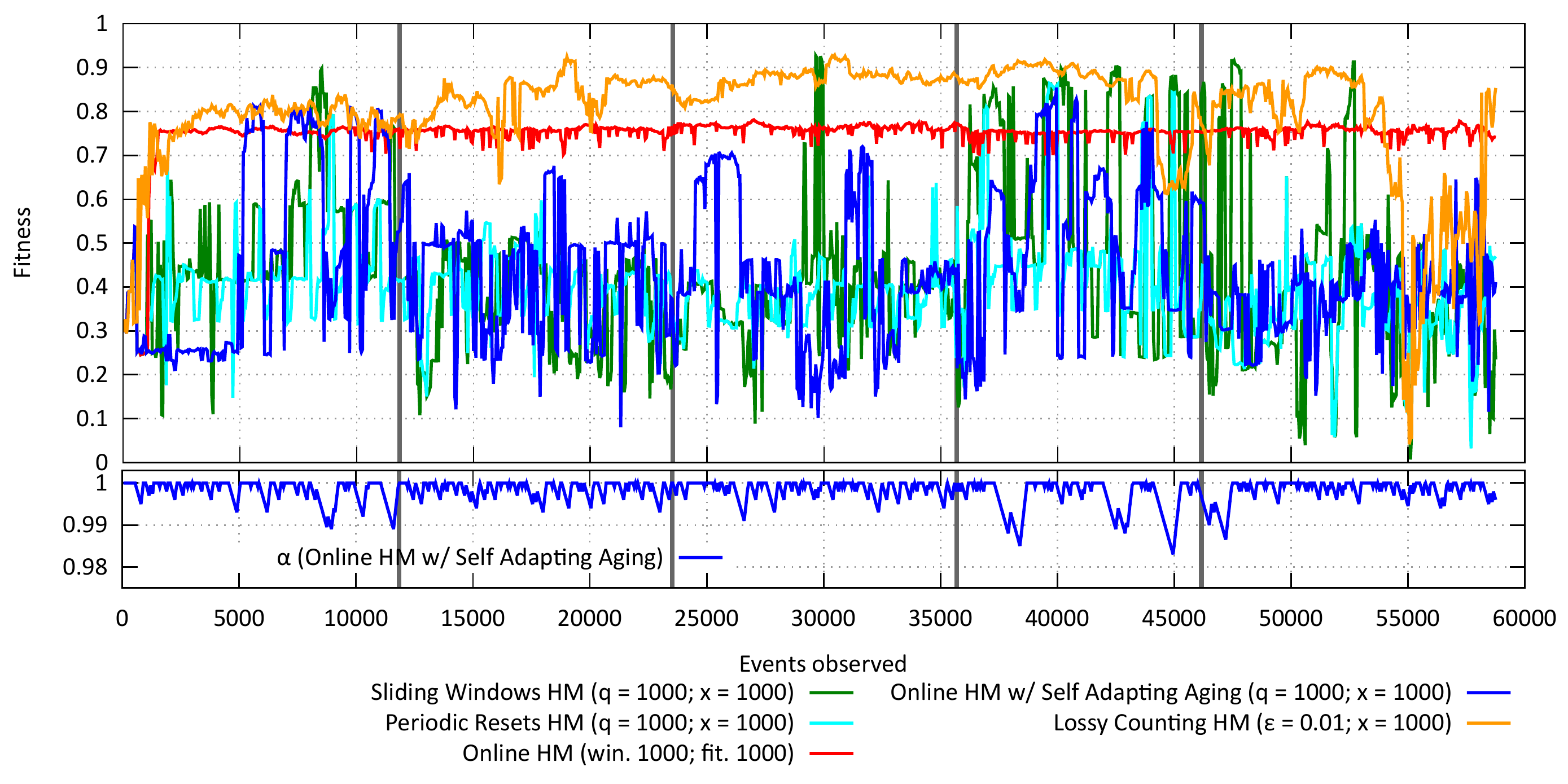} \\[2em]
	\includegraphics[width=\textwidth]{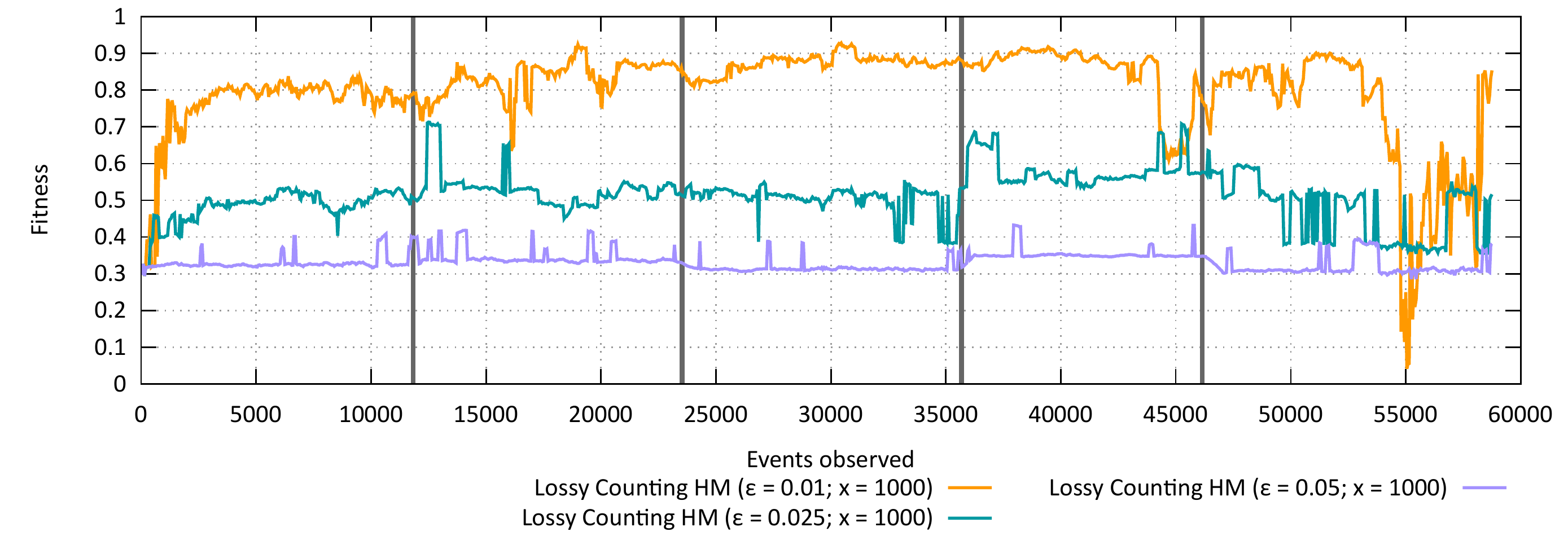}
	\caption{Detailed results of the basic approaches, Online HM, HM with Self Adapting and Lossy Counting (with different configurations) on data of Model~3. Vertical gray lines indicate points where concept drift occur.}
	\label{fig:fitness-performance:insurance}
\end{figure}

The Model 3, with the synthetic example, has been tested with the basic approaches (Sliding Windows and Periodic Resets), the Online HM, the HM with Self Adapting and the Lossy Counting and the results are presented in \cref{fig:fitness-performance:insurance}. In this case, the Lossy Counting and the Online HM outperform the other approaches. Lossy Counting reaches higher fitness values, however Online HM is more stable and seems to better tolerate the drifts. The basic approaches and the HM with Self Adapting, on the other hand, are very unstable; moreover it is interesting to note that the value of $\alpha$, of the HM with Self Adapting, is always close to 1. This indicates that the short stabilities of the fitness values are sufficient to increase $\alpha$, so the updating policy (i.e. the increment/decrement speed of $\alpha$) presented, for this particular case, seems to be too fast. The second graph, on the bottom, presents three runs of the Lossy Counting, with different values for $\epsilon$.
As expected, the lower the value of the accepted error, the better the performances.

Due to the size of this dataset, it is interesting to evaluate the performance of the approaches also in terms of space and time requirements.

\begin{figure}
	\centering
	\includegraphics[width=\textwidth]{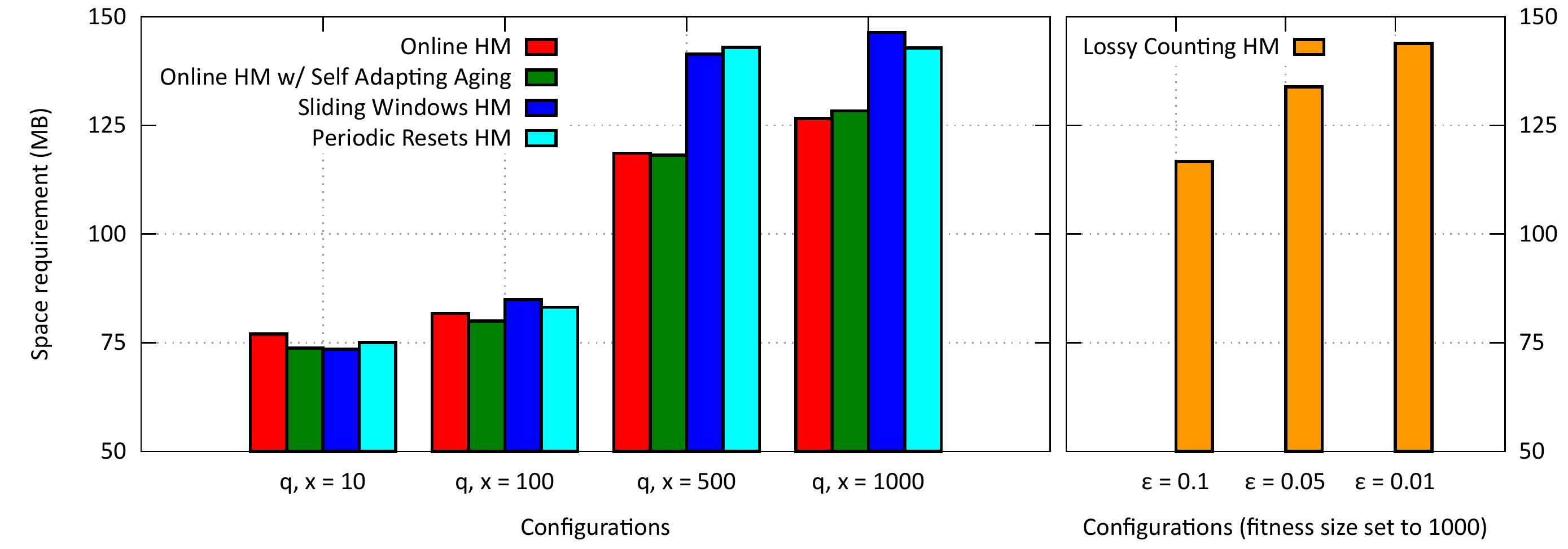}
	\caption{Average memory requirements, in MB, for a complete run over the entire log of Model 3, of the approaches (with different configurations).}
	\label{fig:memory-performance:insurance}
\end{figure}
\cref{fig:memory-performance:insurance} presents the average memory required by the miner during the processing of the entire log. Different configurations are tested, both for the basic approaches with the Online HM and the HM with Self Adapting, and the Lossy Counting algorithm. Clearly, as the windows grow, the space requirement grows too. For what concerns the Lossy Counting, again, as the $\epsilon$ value (accepted error) becomes lower, more space is required. If we pick the Online HM with window 1000 and the Lossy Counting with $\epsilon$ 0.01 (from \cref{fig:fitness-performance:insurance}, both seem to behave similarly) the Online HM consumes less memory: it requires 128.3 MB whereas the Lossy Counting needs 143.8.
\begin{figure}
	\centering
	\includegraphics[width=\textwidth]{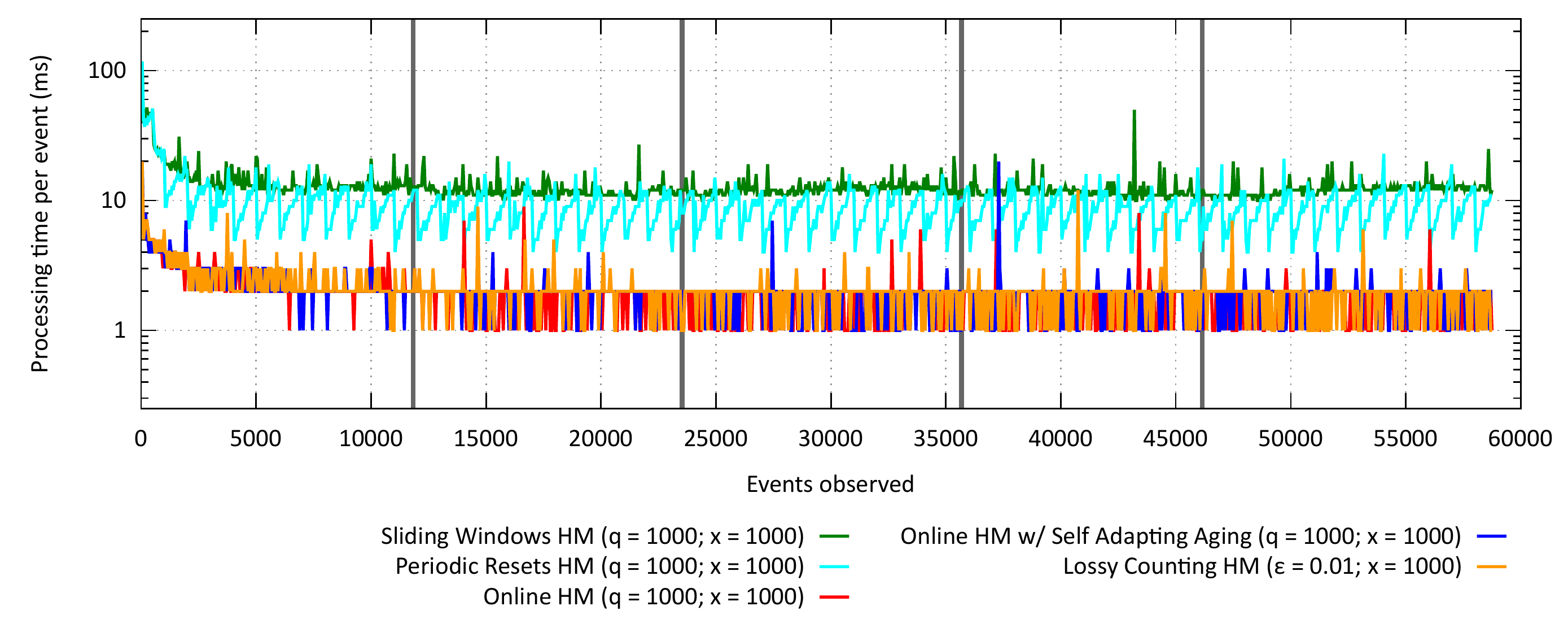} \\[2em]
	\includegraphics[width=\textwidth]{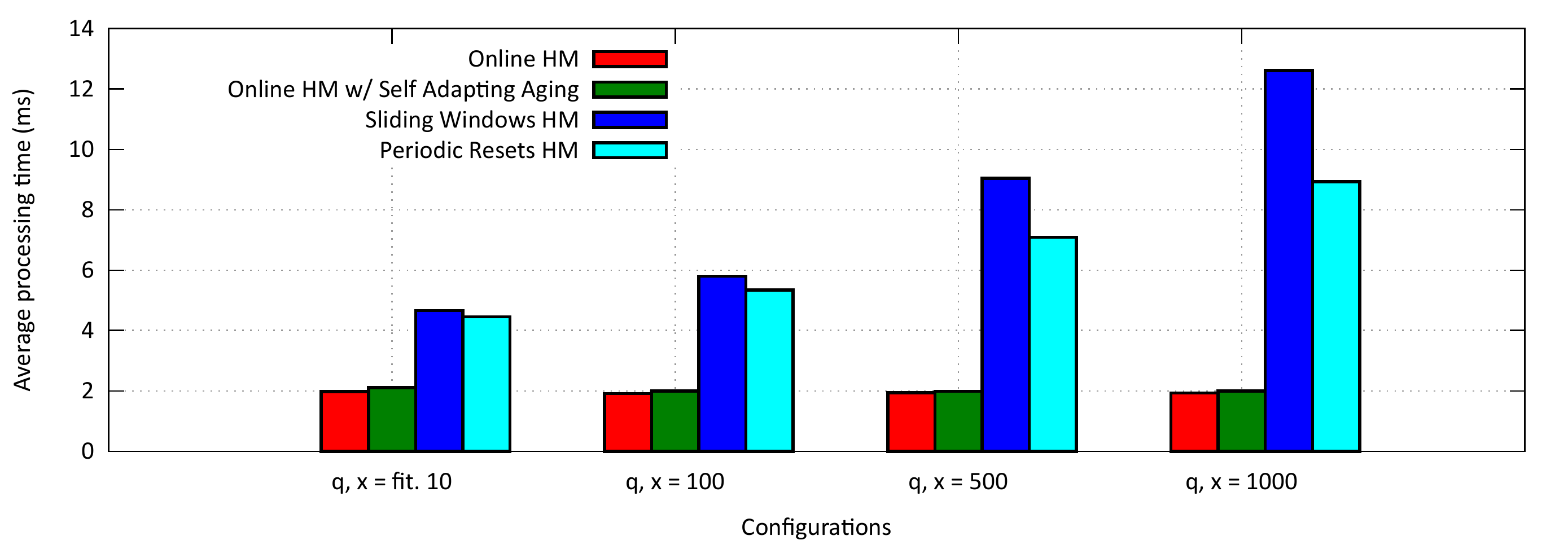}
	\caption{Time performances over the entire log of Model 3. \emph{Top:} time required to process a single event by different algorithms (logarithmic scale). Vertical gray lines indicate points where concept drift occur. \emph{Bottom:} average time required to process an event over the entire log, with different configurations of the algorithms.}
	\label{fig:time-performance:insurance}
\end{figure}
\cref{fig:time-performance:insurance} shows the time performance of different algorithms and different configurations. It is interesting to note, from the chart at the bottom, that the time required by the Online and the Self Adapting is almost independent of the configurations. Instead, the basic approaches need to perform more complex operations: the Periodic Reset has to add the new event and, sometimes, it resets the log; the Sliding Window has to update the log every time a new event is observed.

In order to study the dependence of the storage requirements of Lossy Counting with respect to the error parameter $\epsilon$, we have run experiments on the same log for different values of $\epsilon$, recording the maximum size of the Lossy Counting sets during execution.
Results for $x=1000$ are reported in \cref{fig:size-fitness-precision}. Specifically, the figure compares the maximum size of the generated sets, the average \emph{fitness} value and the average \emph{precision} value. As expected, as the value of $\epsilon$ becomes larger, both the \emph{fitness} value and the sets size quickly decrease. The \emph{precision} value, on the contrary, initially decreases and then goes up to very high values. This indicates an over-specialization of the model to specific behaviors.

\begin{figure}[t]
	\centering
	\includegraphics[width=.8\textwidth]{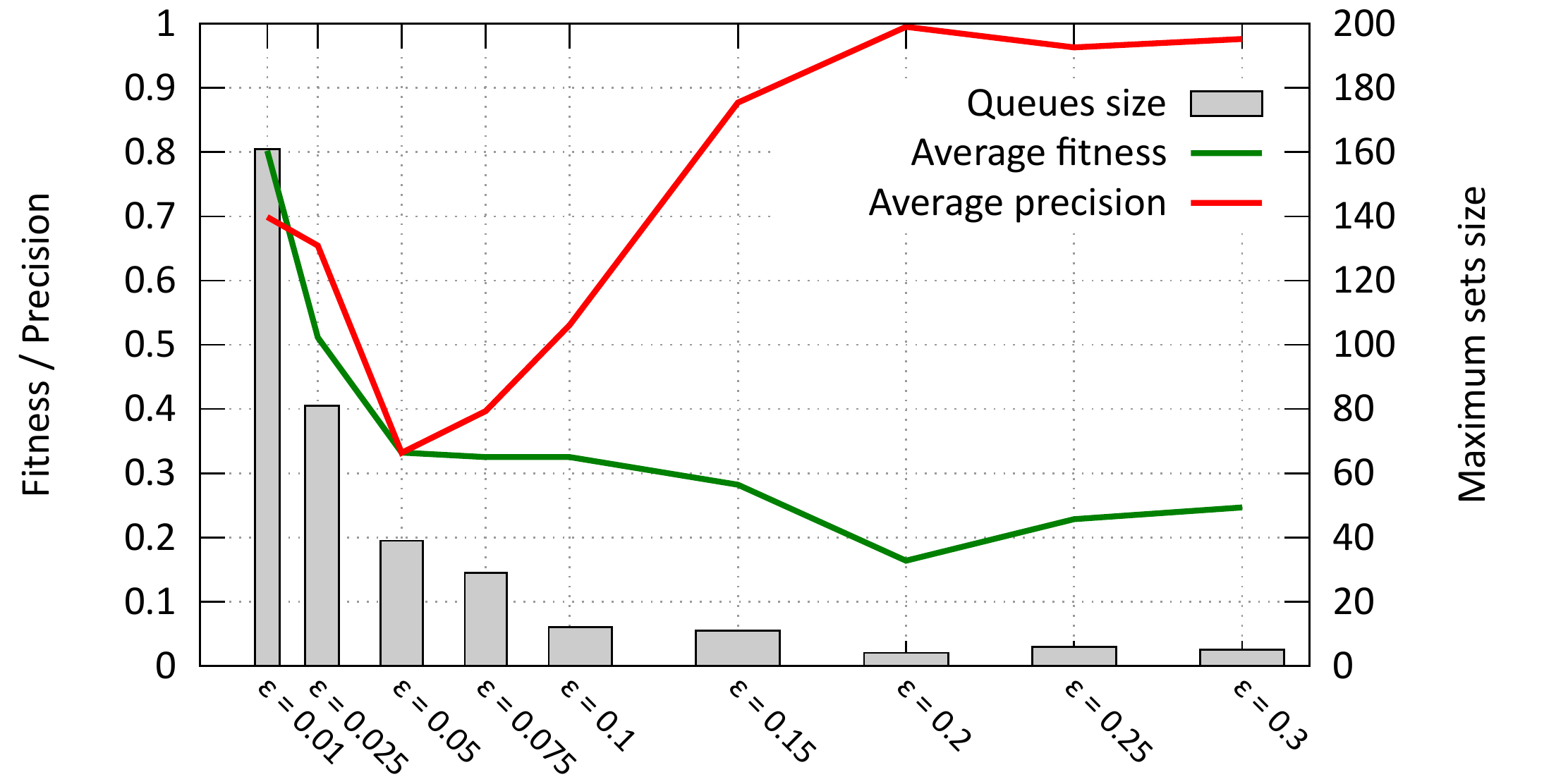}
	\caption{Comparison of the average \emph{fitness}, \emph{precision} and space required, with respect to different values of $\epsilon$ for the Lossy Counting HM executed on the log generated by the Model~3.}
	\label{fig:size-fitness-precision}
\end{figure}

\begin{table}
	\centering
	\begin{tabular}{c|r|cccc}
		\multicolumn{1}{c}{} & \multicolumn{1}{c}{} &
			\begin{footnotesize}\begin{sideways}\begin{minipage}[t]{1.5cm} \textbf{Sliding Window HM} \end{minipage}\end{sideways}\end{footnotesize} &
			\begin{footnotesize}\begin{sideways}\begin{minipage}[t]{1.5cm} \textbf{Lossy Counting HM} \end{minipage}\end{sideways}\end{footnotesize} &
			\begin{footnotesize}\begin{sideways}\begin{minipage}[t]{1.7cm} \textbf{Online HM with Aging} \end{minipage}\end{sideways}\end{footnotesize} &
			\begin{footnotesize}\begin{sideways}\textbf{Online HM} \end{sideways}\end{footnotesize} \\
		\midrule
		\multirow{3}{*}{\textbf{$q=10$}}
		& \emph{Avg. Time (ms)}	& 4.66 & 2.61 & 2.11 & 1.97 \\
		& \emph{Avg. Fitness}	& 0.32 & 0.28 & 0.32 & 0.32 \\
		& \emph{Avg. Precision}	& 0.44 & 0.87 & 0.38 & 0.38 \\
		\midrule
		\multirow{3}{*}{\textbf{$q=100$}}
		& \emph{Avg. Time (ms)}	& 5.79 & 2.85 & 1.99 & 1.91 \\
		& \emph{Avg. Fitness}	& 0.32 & 0.51 & 0.42 & 0.74 \\
		& \emph{Avg. Precision}	& 0.42 & 0.65 & 0.68 & 0.71 \\
		\bottomrule
	\end{tabular}
	\caption{Performance of different approaches with queues/sets size of $q=10$ and $q=100$ elements and $x=1000$. Online HM with Aging uses $\alpha^{1/q} = 0.9$. Time values refer to the average number of milliseconds required to process a single event with respect to Model~3.}
	\label{tbl:times-fitness-precision}
\end{table}

As an additional test, we decide to compare the proposed algorithms under extreme storage conditions which do allow only to retain limited information about the observed events. Specifically,
\cref{tbl:times-fitness-precision} reports the average time required to process a single event, average \emph{fitness} and \emph{precision} values when queues with size $10$ and $100$, respectively, are used. For Lossy Counting we have used an $\epsilon$ value which approximately
requires sets of similar sizes.  Please note that, for this log, a single process trace is longer than $10$ events so, with a queue of $10$ elements it is not possible to keep in queue all the events of a case (because events of different cases are interleaved). From the results it is clear that,
under these conditions, the order of occurrence of the algorithms in the table (column order) is inversely proportional to all the evaluation
criteria (i.e. execution time, \emph{fitness}, \emph{precision}).


\begin{figure}
	\centering
	\includegraphics[width=\textwidth]{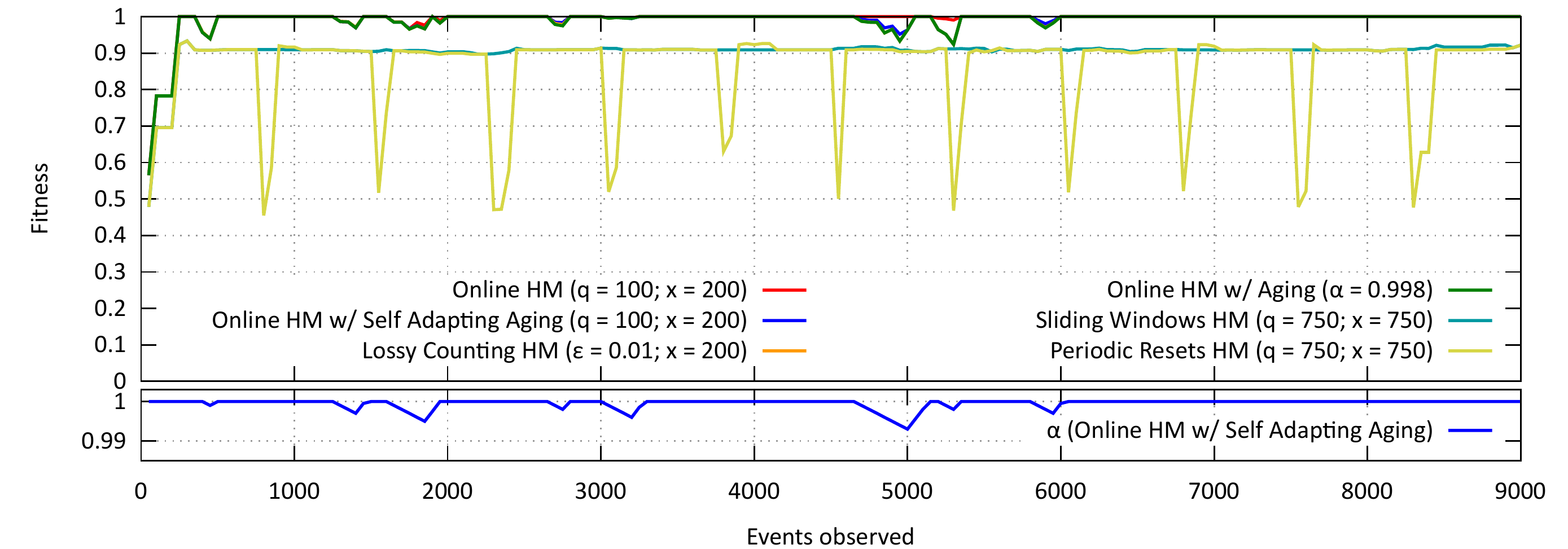}
	\caption{Fitness performance on the real stream dataset by different algorithms.}
	\label{fig:fitness-real}
\end{figure}
The online approaches presented in this work have been tested also against a real dataset and results are presented in \cref{fig:fitness-real}. The reported results refer to 9000 events generated from the document management system, by Siav S.p.A., and run on an Italian bank institute.
The observed process contains 8 activities and is assumed to be stationary. The mining is performed using a queues size of $100$ and, for the fitness computation, the latest $200$ events are considered.
The behavior of the fitness curves seems to indicate that some minor drifts occur.


\begin{figure}[h!]
	\subfloat[Configuration that requires about 100MB. Lossy Counting: $\epsilon: 0.2$, fitness queue size: 200; Online HM: queue size: 500, fitness queue size: 200.]{\includegraphics[width=\textwidth]{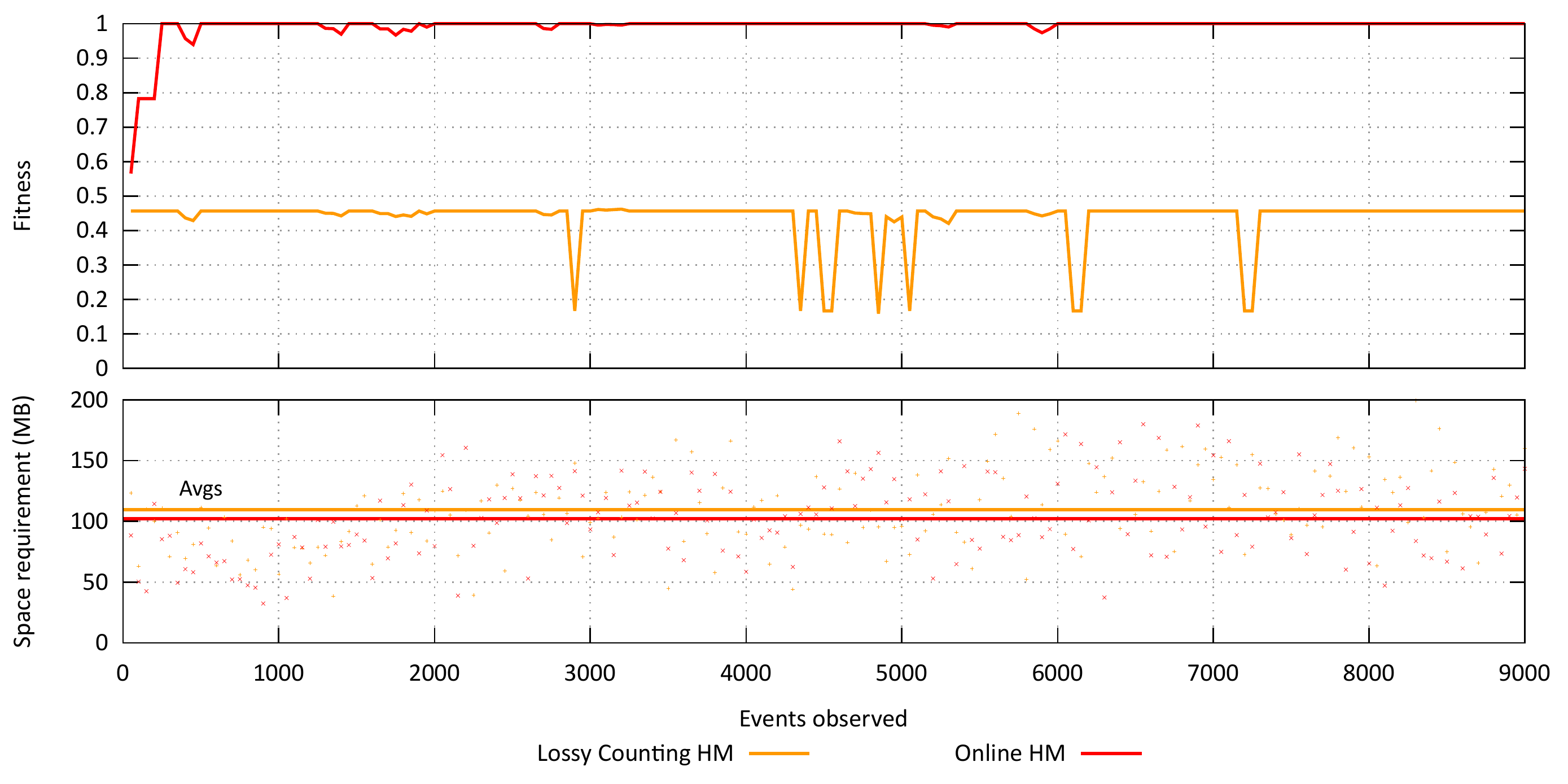}} \\ [1em]
	\subfloat[Configuration that requires about 170MB. Lossy Counting: $\epsilon: 0.01$, fitness queue size: 1000; Online HM: queue size: 1500, fitness queue size: 1000.]{\includegraphics[width=\textwidth]{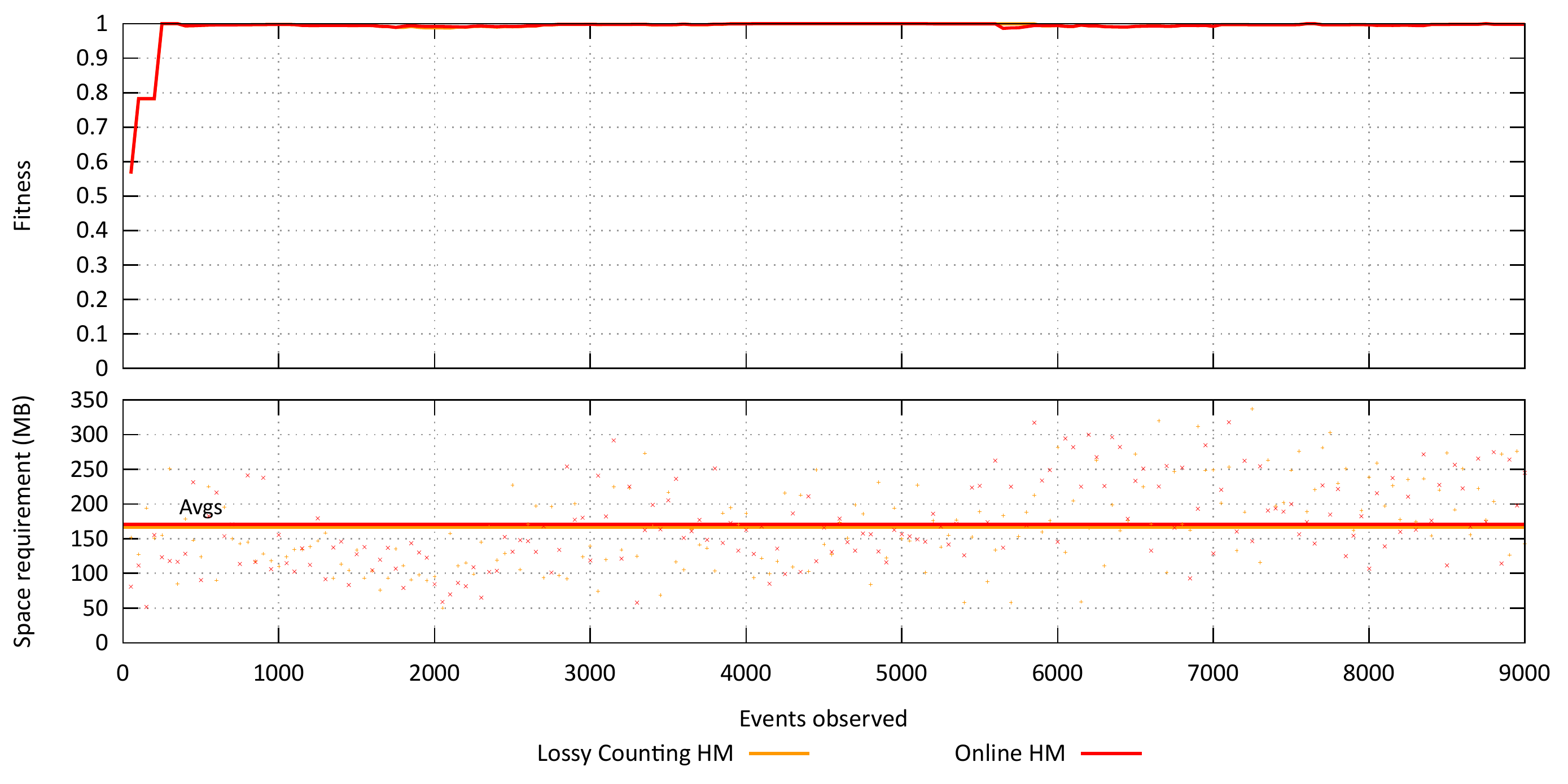}}
	\caption{Performances comparison between Online HM and Lossy Counting, in terms of fitness and memory consumption.}
	\label{fig:memory-comparison-online-lossy}
\end{figure}
As stated before, the main difference between Online HM and Lossy Counting is that, whereas the main parameter of Online HM is the size of the queues (i.e. the maximum space the application is allowed to use), the $\epsilon$ parameter of Lossy Counting cannot control the memory occupancy of the approach. \cref{fig:memory-comparison-online-lossy} proposes two comparisons of the approaches with two different configurations, against the real stream dataset. In particular we defined the two configurations so that the average memory required by Lossy Counting and Online HM are very close. The results presented are actually the average values over four runs of the approaches. Please note that the two configurations validates the fitness against different window sizes (in the first case it contains 200 events, in the second one 1000) and this causes the second configuration to validate results against a larger history.

The top part of the figure presents a configuration that uses, on average, about 100MB. To obtain this performance, several tests have been made and, at the end, for Lossy Counting these parameters have been used: $\epsilon: 0.2$, fitness queue size: 200. For Online HM, the same fitness is used, but the queue size is set to 500. As the plot shows, it is interesting to note that, in terms of fitness, this configuration is absolutely enough for the Online HM approach instead, for Lossy Counting, it is not.
The second plot, at the bottom, presents a different configuration that uses about 170MB. In this case, the error (i.e. $\epsilon$) for Lossy Counting is set to $0.01$, the queue size of Online HM is set to 1500 and, for both, the fitness queue size is set to 1000. In this case the two approaches generate really close results, in terms of fitness.

As final consideration, this empirical evaluation clearly shows that --at least in our real dataset-- both Online HM and Lossy Counting are able to reach very high performances, however the Online is able to better exploit the information available with respect to the Lossy Counting.
In particular, Online HM considers only a finite number of possible observations (depending on the queue size) that, in this particular case, are sufficient to mine the correct model. The Lossy Counting, on the contrary, keeps all the information for a certain time-frame (obtained starting from the error parameter) without considering how many different behaviors are already seen.


\textbf{Note on fitness measure} \hspace{.5em}
The usage of fitness for the evaluation of stream process mining algorithms seems to be an effective choice. However, this might not always be the case: let's consider two very different processes $P'$ and $P''$ and a stream composed of events generated by alternate executions of $P'$ and $P''$.
Under specific conditions, the stream miner will generate a model that contains both $P'$ and $P''$, connected by an initial XOR-split and merged with a XOR-join. This model will have a very high fitness value (it can replay traces from both $P'$ and $P''$), however the mined model is not the one expected, i.e. the alteration in time of $P'$ and $P''$ is not reflected well.

\begin{figure}
	\centering
	\includegraphics[width=\textwidth]{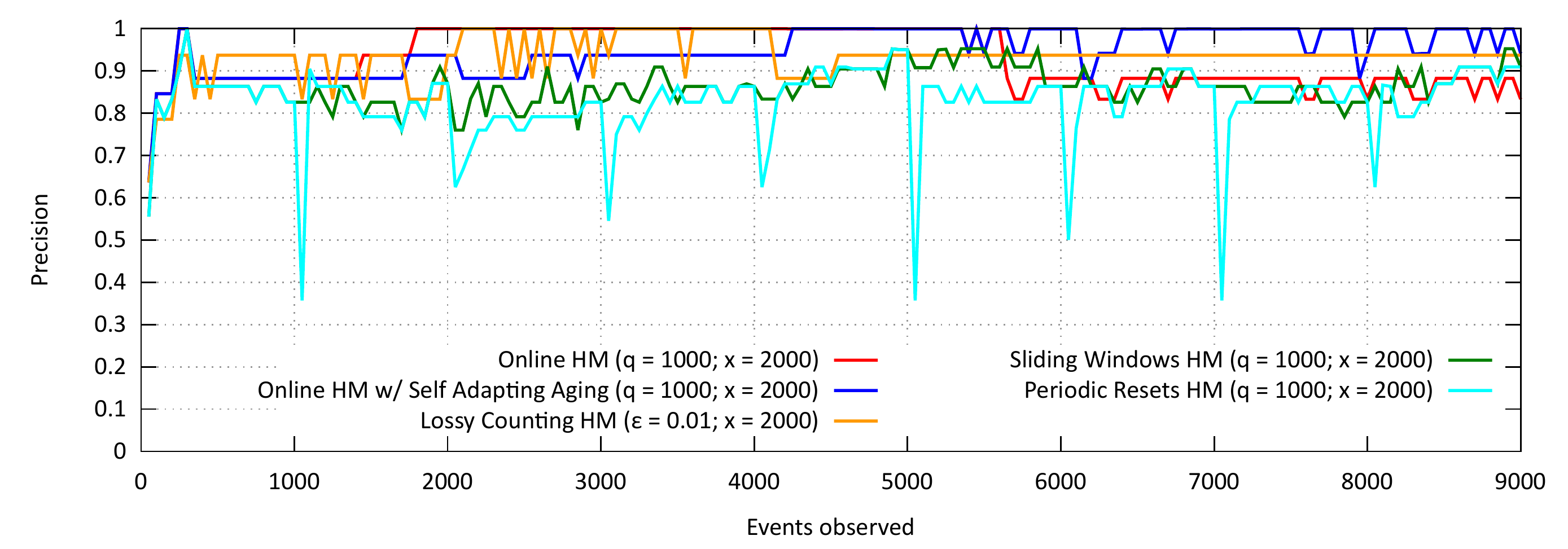}
	\caption{Precision performance on the real stream dataset by different algorithms.}
	\label{fig:precision-real}
\end{figure}

In order to deal with the problem just presented, we propose the performances of some approaches also in terms of ``precision''. This measure is thought to prefer models that describe a ``minimal behavior'' with respect to all the model that can be generated starting from the same log. In particular we used the approach by Mu\~noz-Gama and Carmona described in \cite{Munoz-Gama2010}.
\cref{fig:precision-real} presents the precision calculated for four approaches during the analysis of the dataset of real events. It should not surprise to notice that the stream specific approaches reach very good precision values, whereas the basic approach with periodic reset needs to recompute, every 1000 events, the model from scratch. It is interesting to note that both Online HM and Lossy Counting are not able to reach the top values, whereas the Self adapting one, after some time, reaches the best precision, even if its value fluctuates a bit. The basic approach with sliding window, instead, seems to behave quite nicely, even if the stream specific approaches outperform it.

%% file: section-conclusions.tex
\section{Conclusions and Future Work} \label{sec:conclusions}

In this paper, we addressed the problem of discovering processes for streaming event data having different characteristics, i.e. stationary streams and streams with drift.

First, we considered basic window-based approaches, where the standard Heuristics Miner algorithm is applied to statics logs obtained by using a \emph{moving window} on the stream (we considered two different policies).
Then we introduced a framework for stream process mining which allows the definition of different approaches, all based on the dependencies between activities. These can be seen as online versions of the Heuristics Miner algorithm and differentiate from each other in the way they assign importance to the observed events. The Online HM, an incremental version of the Heuristics Miner, gives the same importance to all the observed events, and thus it is specifically apt to mine stationary streams. HM with Aging gives less importance to older events. This is obtained by weighting the statistics of an event by a factor, the $\alpha$ value, which exponentially decreases with the age of the event. Because of that, this algorithm is able to cope with streams exhibiting concept drift. The choice of the ``right'' value for $\alpha$, however, is difficult and different values for $\alpha$ could also be needed at different times. To address this issue, we finally introduce Heuristics Miner able to automatically
adapt the aging factor on the basis of the detection of concept drift (HM with Self Adapting).
Finally, we adapted a standard approach (Lossy Counting) to our problem.

Experimental results on artificial, synthetic and real data show the efficacy of the proposed algorithms with respect to the basic approaches. Specifically, the Online HM turns out to be a quite stable and performs well for streams, especially when stationary streams are considered, while HM with Self Adapting aging factor and the Lossy Counting seem to be the right choice in case of concept drift. The largest log has been used also for measuring performance in terms of time and space requirements.

As future work, we plan to conduct a deeper analysis of the influence of the different parameters on the presented approaches. Moreover, we plan to extend the current approach also to mine the organizational perspective of the process. Finally, from a process analyst point of view, it may be interesting to not only show the current updated process model, but also report the ``evolution points'' of the process.

%% file: section-heuristics-miner.tex
\section{Heuristics Miner} \label{appendix:HM}

\subsection{Heuristics Miner metrics}

Heuristics Miner (HM) \cite{Weijters2003} is a process mining algorithm that counts various types of frequencies to mine dependency relations among activities represented by logs.

The relation $a >_W b$ holds iff there is a trace $\sigma = \langle t_1, t_2, \dots, t_n \rangle$ and $i \in \{ 1, \dots, n-1 \}$ such that $t_i = a$ and $t_{i+1} = b$. The notation $|a >_W b|$ indicates to the number of times that, in $W$, $a >_W b$ holds (no. of times activity $b$ directly follows activity $a$).

The following subsections present a detailed list of all the formulae required by Heuristics Miner to build a process model.

\subsubsection{Dependency Relations ($\Rightarrow$)}

An edge (that usually represents a dependency relation) between two activities is added if its \emph{dependency measure} is above the value of the \emph{dependency threshold}. This relation is calculated, between activities $a$ and $b$, as:
\begin{equation}
	a \Rightarrow_W b = \frac
		{|a >_W b| - |b >_W a|}
		{|a >_W b| + |b >_W a| + 1}
	\label{eq:dep-threshold}
\end{equation}
The rationale of this rule is that two activities are in a dependency relation if most of times they are in the specifically required order.

\subsubsection{AND/XOR Relations ($\wedge$, $\otimes$)}

When an activity has more than one outgoing edge, the algorithm has to decide whether the outgoing edges are in AND or XOR relation (i.e. the ``type of split'').  Specifically, it has to calculate the following quantity:
\begin{equation}
	a \Rightarrow_W (b \wedge c) = \frac
		{|b >_W c| + |c >_W b|}
		{|a >_W b| + |a >_W c| + 1}
	\label{eq:and-threshold}
\end{equation}
If this quantity is above a given \emph{AND threshold}, the split is an AND-split, otherwise it is considered to be in XOR relation. The rationale, in this case, is that two activities are in an AND relation if most of times they are observed in no specific order (so one before the other and vice versa).

\subsubsection{Long Distance Relations ($\Rightarrow^l$)}

Two activities $a$ and $b$ are in a ``long distance relation'' if there is a  dependency between them, but they are not in direct succession. This relation is expressed by the formula:
\begin{equation}
	a \Rightarrow_W^l b = \frac{|a \ggg_W b|}{|b|+1}
	\label{eq:ld-threshold}
\end{equation}
where $|a \ggg_W b|$ indicates the number of times that $a$ is directly or indirectly (i.e. if there are other different activities between $a$ and $b$) followed by $b$ in the log $W$. If this formula's value is above a \emph{long distance threshold}, then a long distance relation is added into the model.

\subsubsection{Loops of Length one and two}

A loop of length one (i.e. a self loop on the same activity) is introduced if the quantity:
\begin{equation}
	a \Rightarrow_W a = \frac{|a >_W a|}{|a >_W a| + 1}
	\label{eq:loop1}
\end{equation}
is above a \emph{length-one loop threshold}. A loop of length two is considered differently: it is introduced if the quantity:
\begin{equation}
	a \Rightarrow_W^2 b = \frac
		{|a >_W^2 b| + |b >_W^2 a|}
		{|a >_W^2 b| + |b >_W^2 a| + 1}
	\label{eq:loop2}
\end{equation}
is above a \emph{length-two loop threshold}. In this case, the $a >_W^2 b$ relation is observed when $a$ is directly followed by $b$ and then there is $a$ again (i.e. for trace $\sigma = \langle t_1, t_2, \dots, t_n \rangle$ there is an $i \in \{ 1, \dots, n-2 \}$ such that $\sigma \in W$ and $t_i = a$ and $t_{i+1} = b$ and $t_{i+2} = a$).

\subsection{Running Example}

Let's consider the process model shown in \cref{fig:hm:sample-log}. Given the set of activities $\{A,B_1,B_2,C,D\}$, a possible log $W$, with 10 process instances, is:
$$
	W = \left\{ \langle A, B_1, B_2, C, D\rangle^5 \ ; \ \langle A, B_2, B_1, C, D \rangle^5 \right\}
$$
Please note that the notation $\langle \cdots \rangle^n$ indicates $n$ case following of the same sequence.
Such log can be generated starting from executions of the process model of \cref{fig:hm:sample-log}.

\begin{figure}
	\centering
	\includegraphics[width=.7\textwidth]{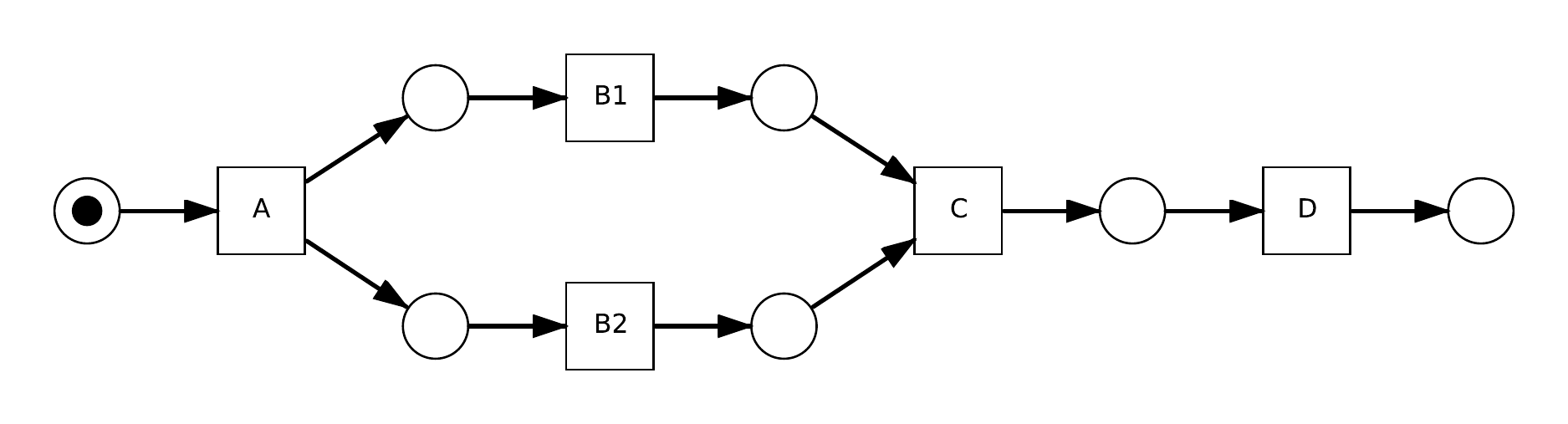}
	\caption{Example of a possible process model that generates the log $W$.}
	\label{fig:hm:sample-log}
\end{figure}

In the case reported in figure, the main measure (dependency relation) builds the following relation:
$$
\bordermatrix{ %
		&	A	& B_1	& B_2	& C		& D     \cr %
	A	& 0	& 0.8\overline{3}	& 0.8\overline{3}	& 0	& 0 \cr %
	B_1	& -0.8\overline{3}	& 0	& 0	& 0.8\overline{3}	& 0 \cr %
	B_2	& -0.8\overline{3}	& 0	& 0	& 0.8\overline{3}	& 0 \cr %
	C	& 0	& -0.8\overline{3}	& -0.8\overline{3}	& 0	& 0.9\overline{09} \cr %
	D	& 0	& 0	& 0	& -0.9\overline{09}	& 0 \cr %
} %
$$
Starting from this relation and considering -- for example -- a value $0.9$ for the \emph{dependency threshold}, it is possible to identify the complete set of dependencies, including the split from activity $A$ to $B_1$ and $B_2$. In order to identify the type of the split it is necessary to use the AND measure (\cref{eq:and-threshold}):
$$
	A \Rightarrow_W (B_1 \wedge B_2) = \frac{5+5}{5+5+1} = 0.9\overline{09}
$$

So, considering --- for example --- an \emph{AND-threshold} of $0.1$, the type of the split is set to AND.
In the ProM implementation, the default value for \emph{dependency threshold} is $0.9$, and for the \emph{AND-threshold} it is $0.1$.

%% file: section-error-bounds.tex
\section{Error Bounds on Online Heuristics Miner} \label{sec:error-bound}

If we assume a stationary stream, i.e. a stream where the distribution of events does not change with time (no concept drift), then it is possible to give error bounds on the measures computed by the online version of Heuristics Miner.

In fact, let consider an execution of the online Heuristics Miner on the stream $S$.
Let $Q_\mathcal{A}(t)$, $Q_\mathcal{C}(t)$, and $Q_\mathcal{R}(t)$ be the content of the queues used by \cref{online_alg} at time $t$.
Let $\mathit{case}_{\mathit{overlap}}(t) = \{c \in \mathcal{C} \mid t_{\mathit{start}}(c) \leq t \wedge t_{\mathit{end}}(c ) \geq t\}$ be the set of cases that are \emph{active} at time $t$; $\Delta_c = \max_t | \mathit{case}_{\mathit{overlap}}(t)|$; $n_c(t)$ be the cumulative number of cases which have been removed from  $Q_\mathcal{C}(t)$ during the time interval $[0,t]$; and $nc(t)=|Q_\mathcal{C}(t)|+n_c(t)$.
Given two activities $a$ and $b$, let $\rho_{ab}\in [0,\xi_{ab}]$ be the random variable reporting the number of successions $(a,b)$ contained in a randomly selected trace in $S$.
With $\mathcal{A}_S$ and $\mathcal{R}_S$ we denote the set of activities and successions, respectively, observed for the entire stream S. Then it is possible to state the following theorem:

\begin{theorem}[Error bounds]
Let $(a\Rightarrow_S b)$, $a\Rightarrow_S (b\wedge c)$, 
be the measures computed by the Heuristics Miner algorithm on a time-stationary stream $S$, and
$(a\Rightarrow_{S_0^t} b)$, $a\Rightarrow_{S_0^t} (b\wedge c)$, 
be the measures computed at time $t$ by the online version of the Heuristics Miner algorithm on the stream $S$.
If $\max_A \geq |A_S|$, $\max_R \geq |R_S|$, $\max_C \geq \Delta_c$, then with probability $1-\delta$ the following bounds hold:
\begin{multline*}
	(a\Rightarrow_S b)\left(\frac{E[ \rho_{ab}+\rho_{ba} ]}{E[\rho_{ab}+\rho_{ba}]+\epsilon_{ab}(t)+\frac{1}{nc(t)}} \right) - \\
	\frac{\epsilon_{ab}(t)}{E[\rho_{ab}+\rho_{ba}]+\epsilon_{ab}(t)+\frac{1}{nc(t)}}
	\leq (a\Rightarrow_{S_0^t} b)
\end{multline*}
\begin{multline*}
	(a\Rightarrow_{S_0^t} b) \leq
	(a\Rightarrow_S b)\left(\frac{E[ \rho_{ab}+\rho_{ba} ]}{E[\rho_{ab}+\rho_{ba}]-\epsilon_{ab}(t)+\frac{1}{nc(t)}} \right) + \\
	\frac{\epsilon_{ab}(t)}{E[\rho_{ab}+\rho_{ba}]-\epsilon_{ab}(t)+\frac{1}{nc(t)}}
\end{multline*}

\noindent
And, similarly, for $a\Rightarrow (b\wedge c)$:
\begin{multline*}
	(a\Rightarrow_S (b\wedge c))\left(\frac{E[ \rho_{bc}+\rho_{cb} ]}{E[\rho_{ab}+\rho_{ac}]+\epsilon_{abc}(t)+\frac{1}{nc(t)}} \right) - \\
	\frac{\epsilon_{bc}(t)}{E[\rho_{ab}+\rho_{ac}]+\epsilon_{abc}(t)+\frac{1}{nc(t)}}
	\leq (a\Rightarrow_{S_0^t} (b\wedge c))
\end{multline*}
\begin{multline*}
	(a\Rightarrow_{S_0^t} (b\wedge c)) \leq
	(a\Rightarrow_S (b\wedge c))\left(\frac{E[ \rho_{bc}+\rho_{cb} ]}{E[\rho_{bc}+\rho_{cb}]-\epsilon_{abc}(t)+\frac{1}{nc(t)}} \right) + \\
	\frac{\epsilon_{bc}(t)}{E[\rho_{ab}+\rho_{ac}]-\epsilon_{abc}(t)+\frac{1}{nc(t)}}
\end{multline*}

\noindent
where $\forall d,e,f \in A_S,\ \epsilon_{de}(t) = \sqrt{\frac{(\xi_{de}+\xi_{ed})^2 \ln(2/\delta)}{2nc(t)}},\ \epsilon_{def}(t) = \sqrt{\frac{(\xi_{de}+\xi_{df})^2 \ln(2/\delta)}{2nc(t)}}$, and $E[x]$ is the expected value of $x$.
\end{theorem}


\begin{proof}
Let consider the Heuristics Miner definition $(a\Rightarrow_S b)=\frac{|a>_S b|-|b>_S a|}{|a>_S b|+|b>_S a|+1}$ (as presented in \cref{eq:dep-threshold}). Let $N_c$ be the number of cases contained in $S_0^t$, then
$$
	(a\Rightarrow_{S_0^t} b) =
		\frac{| a >_{S_0^t} b | - | b >_{S_0^t} a |}{|a>_{S_0^t} b|+|b>_{S_0^t} a|+1} =
		\frac{\frac{| a >_{S_0^t} b | - | b >_{S_0^t} a |}{N_c}}{\frac{|a>_{S_0^t} b|+|b>_{S_0^t} a|}{N_c}+\frac{1}{N_c}}
$$
and
$$
	(a\Rightarrow_S b)=\lim_{N_c\rightarrow +\infty} \frac{\frac{|a>_{S_0^t} b|-|b>_{S_0^t} a|}{N_c}}{\frac{|a>_{S_0^t} b|+|b>_{S_0^t} a|}{N_c}+\frac{1}{N_c}}= \frac{E[ \rho_{ab}-\rho_{ba} ]}{E[\rho_{ab}+\rho_{ba}]}.
$$
We recall that $\overline{X}=\frac{|a>_{S_0^t} b|-|b>_{S_0^t} a|}{N_c}$ is the mean of the random variable $X = (\rho_{ab}-\rho_{ba})$ computed over $N_c$ independent observations, i.e. traces, and that  $X\in [-\xi_{ba},\xi_{ab}]$.
We can then use the \emph{Hoeffding} bound \cite{Hoeffding1963} that states that, with probability $1-\delta$
$$
	\left| \overline{X}-E[X] \right| < \epsilon_X = \sqrt{\frac{r_X^2 \ln\left(\frac{2}{\delta}\right)}{2N_c}},
$$
\noindent where $r_X$ is the range of $X$, which in our case is $r_X=(\xi_{ab}+\xi_{ba})$.

By using the \emph{Hoeffding} bound also for the variable $Y=(\rho_{ab}+\rho_{ba})$, we can state that with probability $1-\delta$
$$
	\frac{E[X] - \epsilon_X}{E[Y]+\epsilon_Y+\frac{1}{N_c}} \leq \frac{\overline{X}}{\overline{Y}+\frac{1}{N_c}} = (a\Rightarrow_{S_0^t} b),
$$
\noindent which after some algebra can be rewritten as
\[\frac{E[X]}{E[Y]}\left(\frac{E[Y]}{E[Y]+\epsilon_{Y}+\frac{1}{N_c}} \right) - \frac{\epsilon_{X}}{E[Y]+\epsilon_{Y}+\frac{1}{N_c}}\leq  (a\Rightarrow_{S_0^t} b)\]
By observing that $(a\Rightarrow_{S} b) = \frac{E[X]}{E[Y]}$, $r_X=r_Y=(\xi_{ab}+\xi_{ba})$, and that at time $t$, under the theorem hypotheses, no information is removed from the queues and $N_c = nc(t)$, the
first bound is proved.
The second bound can be proved starting from
\[(a\Rightarrow_{S_0^t} b) \leq \frac{E[X] + \epsilon_X}{E[Y]-\epsilon_Y+\frac{1}{N_c}}.
\]

The last two bounds can be proved in a similar way by considering $X=(\rho_{bc}+\rho_{cb})\in [0, \xi_{bc} + \xi_{cb}]$ and $Y=(\rho_{ab}+\rho_{ac})\in [0, \xi_{ab} + \xi_{ac}]$, which leads to $ \epsilon_X = \sqrt{\frac{(\xi_{bc} + \xi_{cb})^2ln(2/\delta)}{2N_c}}$ and $ \epsilon_Y = \sqrt{\frac{(\xi_{ab} + \xi_{ac})^2ln(2/\delta)}{2N_c}}$.
\end{proof}

Similar bounds can be obtained also for the other measures computed by Heuristics Miner. From the bounds it is possible to see that, with the increase of the number of observed cases $nc(t)$, both $\frac{1}{nc(t)}$ and the errors $\epsilon_{ab}(t)$ and $\epsilon_{abc}(t)$ go to $0$ and the measures computed by the online version of Heuristics Miner consistently converge to the  ``right'' values.